\documentclass[11pt]{article}

\newcommand{\by}{\mathbf{y}}
\usepackage{placeins}
\usepackage[margin=1in]{geometry}
\usepackage[T1]{fontenc}
\usepackage[utf8]{inputenc}
\usepackage{lmodern}
\usepackage{microtype}
\usepackage{amsmath,amssymb,amsfonts,amsthm,mathtools,bm}
\usepackage{aliascnt}
\usepackage{booktabs,array,tabularx,longtable}
\usepackage{enumitem}
\usepackage{graphicx}
\usepackage{verbatim}
\usepackage{rotating}
\usepackage{float}
\usepackage{xcolor}

\usepackage[round,authoryear]{natbib}
\usepackage[colorlinks=true,allcolors=blue]{hyperref}
\usepackage[nameinlink,noabbrev]{cleveref}

\newtheorem{proposition}{Proposition}[section]
\newaliascnt{theorem}{proposition}
\newtheorem{theorem}[theorem]{Theorem}
\aliascntresetthe{theorem}
\newaliascnt{corollary}{proposition}
\newtheorem{corollary}[corollary]{Corollary}
\aliascntresetthe{corollary}
\newaliascnt{remark}{proposition}

\aliascntresetthe{remark}
\theoremstyle{definition}
\newaliascnt{definition}{proposition}

\aliascntresetthe{definition}
\crefname{theorem}{theorem}{theorems}
\Crefname{theorem}{Theorem}{Theorems}
\crefname{corollary}{corollary}{corollaries}
\Crefname{corollary}{Corollary}{Corollaries}

\newcommand{\R}{\mathbb{R}}
\newcommand{\E}{\mathbb{E}}
\newcommand{\Var}{\operatorname{Var}}
\newcommand{\Cov}{\operatorname{Cov}}
\newcommand{\Corr}{\operatorname{Corr}}
\newcommand{\Prb}{\mathbb{P}}
\newcommand{\Dir}{\operatorname{Dirichlet}}
\newcommand{\Cat}{\operatorname{Categorical}}
\newcommand{\diag}{\operatorname{diag}}
\newcommand{\KL}{\operatorname{KL}}
\newcommand{\trans}{\mathsf{T}}
\newcommand{\bpi}{\boldsymbol{\pi}}
\newcommand{\bp}{\boldsymbol{p}}
\newcommand{\bY}{\mathbf{Y}}
\newcommand{\bx}{\mathbf{x}}
\newcommand{\be}{\boldsymbol{e}}
\newcommand{\bmu}{\boldsymbol{\mu}}
\newcommand{\btheta}{\boldsymbol{\theta}}
\newcommand{\bSigma}{\boldsymbol{\Sigma}}
\newcommand{\cH}{\mathcal{H}}
\newcommand{\simplex}{\Delta_{K-1}}

\title{\textbf{A Dirichlet Mixed-Membership Model for Exact Multivariate Distributional Credibility}}
\author{
    Sebasti{\'a}n Calcetero Vanegas\thanks{Independent Researcher, Toronto, Canada. Email: \texttt{{sebastian.calcetero@alumni.utoronto.ca}}}
    \and
    Ian Weng Chan\thanks{Division of Finance, Nanyang Business School, Nanyang Technological University, Singapore. Email: \texttt{ianweng.chan@ntu.edu.sg}}
}
\date{\today}

\begin{document}
\maketitle

\begin{abstract}
Credibility theory combines individual experience with portfolio information for insurance pricing, but classical formulations focus primarily on conditional means and expected premiums. We propose a Dirichlet mixed-membership model (DMMM) for multivariate distributional credibility. Policyholder risk is represented by a stable composition over latent risk classes: baseline characteristics determine its a priori assessment, while repeated multivariate experience progressively updates its a posteriori assessment. Support-specific expert distributions accommodate heterogeneous outcomes, with dependence induced through the shared latent structure. The resulting posterior predictive distribution can be written exactly as a convex combination of portfolio and experience components, extending the familiar credibility-factor structure beyond the conditional mean. We study the framework through a simulation experiment and a vehicle telematics application. The simulation shows that relatively simple experts can capture complex multivariate insurance distributions and learn policyholder-specific risk as experience accumulates. In the telematics application, recent claims and driving behaviour provide complementary information for future claim-frequency prediction, allowing similar policyholders to receive different experience-rated assessments. The DMMM therefore provides a flexible and interpretable framework for combining heterogeneous insurance experience while retaining the structure of classical credibility.
\end{abstract}


\newpage
\section{Introduction}\label{sec:intro}

Credibility theory provides a coherent framework for incorporating a policyholder's emerging experience into risk assessment and pricing. Individual histories are directly relevant to future risk but are often too short or sparse to support reliable estimation on their own; portfolio information is more stable but less specific to the individual. Credibility reconciles these two sources by borrowing strength from the portfolio while allowing accumulated experience to progressively influence the policyholder's assessment. This principle underlies experience rating and a posteriori ratemaking and remains central to actuarial practice. The classical foundations were established by \citet{Buhlmann1967}, \citet{BuhlmannStraub1970}, \citet{Jewell1974}, and \citet{Jewell1975}; see also \citet{Norberg1979} and \citet{BuhlmannGisler2005} for further developments and broader treatments of credibility theory.

A key strength of classical credibility is its combination of statistical interpretation and analytical simplicity. A hierarchical representation of policyholder heterogeneity leads naturally to a balance between portfolio information and individual experience, often expressed as a convex combination governed by a credibility factor. Much of the classical development, however, centres on conditional means and expected premiums, with claim history providing the principal source of a posteriori information. Several extensions have broadened this framework. Multivariate credibility permits information to be shared across claim frequency and severity, multiple lines of business, or several related responses \citep{Frees2003,EnglundEtAl2008,FreesValdez2008}, while distributional credibility seeks to update the underlying risk distribution rather than only its first moments; see, for example, \citet{CaiEtAl2015}. Related Bayesian approaches have pursued greater flexibility in both latent-risk modelling and posterior prediction, allowing experience to inform a richer predictive distribution than is available from moment-based credibility alone \citep{CheungEtAl2021}.

Bringing multivariate and distributional credibility together is less straightforward. Doing so requires a model that can accommodate responses on different supports, represent dependence across them, and allow the joint predictive distribution to change flexibly as experience accumulates, while still making clear how portfolio and individual information are being combined. Richer latent and dependence structures can also lead to substantial computational demands and may obscure the transparent portfolio-versus-experience interpretation that makes credibility attractive. Computational approaches based on surrogate modelling and population sampling have been proposed to make experience rating under complex hierarchical models more practical \citep{CalceteroVanegasEtAl2024,CalceteroVanegasLin2024}, but computational tractability alone does not resolve these modelling and interpretability issues. The central challenge is therefore to achieve multivariate and distributional flexibility without giving up a tractable and interpretable credibility update.

These developments are particularly relevant as insurers observe increasingly rich forms of policyholder information. Usage-based motor insurance (UBI) provides a prominent example: telematics systems record mileage, driving times, acceleration and braking behaviour, turning patterns, and related measures that may reveal aspects of risk not apparent from a short claim history. Because these signals are observed while the policy is in force, they naturally constitute additional a posteriori information. \citet{DenuitGuillenTrufin2019} developed a multivariate credibility framework combining claim frequencies and behavioural signals, and related longitudinal and dynamic models have subsequently been studied by \citet{CorradinEtAl2022}, \citet{HenckaertsAntonio2022}, and \citet{YanezEtAl2025}. These contributions demonstrate the value of borrowing information across outcomes, but flexible approaches that allow heterogeneous experience to update a complete multivariate predictive distribution remain comparatively limited.

Motivated by this gap, we develop a Dirichlet mixed-membership model (DMMM) based on a finite collection of latent risk classes. Rather than assigning each policyholder permanently to one class, the underlying policyholder risk is represented by a stable composition over these classes. Baseline attributes determine the portfolio-level assessment of this composition, while different periods may express different classes and repeated observations progressively refine the insurer's posterior assessment of the individual composition. Each class is described by a multivariate expert distribution, so different response components may be assigned support-appropriate models for counts, losses, severities, proportions, or continuous behavioural signals. Dependence across outcomes is induced through their shared latent risk classes, avoiding the need to impose a common marginal family, a separate copula, or a rigid parametric covariance structure. By increasing the number of classes, the model can represent increasingly rich multivariate conditional distributions.

Importantly, the additional flexibility does not come at the expense of the classical credibility interpretation. Posterior risk-class weights admit an exact convex-combination decomposition between a portfolio classification and an experience-based classification, with the balance governed by the strength of the prior and the amount of observed history. The same structure carries through to the complete future predictive distribution. Consequently, the expected-loss premium retains the familiar form of a credibility-weighted combination of a priori and experience-based assessments. More generally, the full predictive distribution provides the basis for computing tail probabilities, quantiles, premiums under nonlinear principles, and other actuarial risk measures when required. The DMMM therefore extends the classical portfolio-versus-experience logic from expected premiums to flexible multivariate predictive distributions.

The DMMM is also closely related to the literature on mixture models for credibility and a posteriori ratemaking. Finite-mixture models represent heterogeneous insurance portfolios through latent subpopulations, with individual experience changing the relative importance of mixture components and hence the resulting experience-rated premium \citep{TzougasVrontosFrangos2014,JahanbaniPayandehMasoumifard2024,AbedPayandehNajafabadi2026}. Bayesian approaches introduce further flexibility through finite- or infinite-mixture structures for the underlying risk or claim distribution. These include Bayesian mixture credibility and nonparametric formulations based on Dirichlet processes, which yield flexible credibility premiums and, in some cases, complete posterior predictive distributions \citep{LauSiuYang2006,SiuYang2009,HongMartin2017,ZhangHuangWu2017}. More recently, latent-state and mixture-of-experts models have incorporated covariates and longitudinal heterogeneity in different ways. \citet{Verschuren2022} considers dynamically evolving latent risk classes for joint frequency--severity experience rating, while \citet{TseungChanFungBadescuLin2023} combines covariate-based mixture-of-experts classification with policyholder-specific random effects.

The DMMM uses these ingredients differently. Its persistent policyholder risk is the composition over latent classes itself, rather than a single fixed or dynamically evolving class. Period-specific classes are repeatedly drawn from this composition, so heterogeneous multivariate outcomes jointly provide information about the same underlying individual risk. The role of the Dirichlet distribution also differs from that in Bayesian nonparametric mixture models: in the canonical DMMM, population-level parameters are estimated frequentistically, while the Dirichlet distribution describes heterogeneity in policyholder-specific risk compositions rather than acting as a prior on an unknown claim or parameter distribution. In \citet{TseungChanFungBadescuLin2023}, persistent heterogeneity is instead introduced through policyholder-specific random effects in the gating structure and posterior inference requires numerical approximation. The DMMM therefore differs in where the persistent latent risk enters the model and in how that risk is updated through repeated experience in closed form.

We study the framework in two complementary numerical settings. First, a simulation experiment considers joint claim frequency, severity, and losses under a data-generating mechanism deliberately more complex than the fitted expert families. It examines the model's ability to recover relevant marginal and dependence features and to learn policyholder-specific predictive distributions under model misspecification. Second, we consider claims and telematics using the synthetic portfolio of \citet{SoBoucherValdez2021}, extended to a controlled longitudinal setting motivated by the SEAT framework of \citet{SoJeong2025}. The application examines whether behavioural experience provides predictive information beyond conventional rating characteristics and past claims, and how the resulting experience updates compare with multivariate Gaussian random-effects and linear credibility approaches.

Taken together, the paper makes three related contributions. First, we develop a multivariate distributional credibility framework in which heterogeneous outcomes jointly update a stable policyholder-specific composition over latent risk classes. Second, this construction yields an exact credibility decomposition for the complete posterior predictive distribution, while retaining sufficient flexibility to approximate complex multivariate risks and learn individual predictive distributions from repeated experience. Third, we develop a practical estimation and inference procedure and study the framework both under controlled model misspecification and in a telematics application, where claims and behavioural information jointly enter experience rating. The contribution lies in combining these features within a single model that retains the transparent portfolio-versus-experience structure of classical credibility, in a mathematically tractable way.

The remainder of the paper is organised as follows. \Cref{sec:model} introduces the model, derives its posterior and predictive quantities, and discusses its flexibility and learning properties. \Cref{sec:estimation} develops estimation and implementation through an EM-type procedure with an optional response-weighting mechanism for auxiliary signals. \Cref{sec:dmm-simulation-study,sec:application} present the simulation experiment and telematics application, respectively. \Cref{sec:conclusion} concludes.

\section{The Dirichlet Mixed-Membership Model (DMMM)}\label{sec:model}

The classical static credibility framework
\citep{BuhlmannGisler2005}
begins with a latent risk attached to each policyholder. In generic notation, one may write
\begin{equation}\label{eq:classical_credibility_framework}
\boldsymbol\zeta_i\sim G,
\qquad
\bY_{it}\mid\boldsymbol\zeta_i,\be_{it}
\sim F(\cdot\mid\boldsymbol\zeta_i,\be_{it}),
\qquad t=1,\ldots,T_i.
\end{equation}
The risk $\boldsymbol\zeta_i$ is fixed for policyholder $i$, but differs across the portfolio according to $G$. Given the risk, the observations are independent across periods after accounting for known period-level information (e.g., exposure) $\be_{it}$. Experience rating, or credibility, consists of updating the distribution of $\boldsymbol\zeta_i$ using the observed claim history, averaging $F(\cdot\mid\boldsymbol\zeta_i,\be_{i,T_i+1})$ over that updated distribution, and then using the resulting predictive distribution to determine premiums.

The proposed model specifies the two abstract objects in \cref{eq:classical_credibility_framework} in a way that is both flexible and operationally interpretable. Insurance companies often summarise heterogeneous portfolios through a finite set of internal risk classes. Such classes are useful because they can be monitored, explained, and translated into premiums or underwriting actions. We therefore take the latent risk to be a class-composition vector,
$\boldsymbol\zeta_i=\bp_i$. Its $k$th coordinate, $p_{ik}$, is the long-run proportion of periods in which policyholder $i$ expresses risk class $k$. The portfolio distribution $G$ is modelled by a Dirichlet distribution centred on a baseline classification obtained from observable attributes.

This construction does not force a policyholder to remain in one class in every period. Conditional on $\bp_i$, a period-specific class is drawn, after which the complete multivariate outcome is generated from the corresponding class distribution. Thus, the classes describe recognisable patterns of insurance experience, whereas $\bp_i$ is the stable policyholder-level risk that credibility seeks to learn. New observations do not change the underlying $\bp_i$; they change the insurer's posterior information about it.

\subsection{Mathematical formulation}\label{subsec:formulation}

Let $i=1,\ldots,n$ index policyholders and $t=1,\ldots,T_i$ index observation periods. The history length $T_i$ may differ across policyholders. Let $\bx_i\in\R^p$ contain baseline policyholder attributes used for the prior risk classification. Let $\be_{it}\in\mathcal E$ collect period-level information entering the response distributions, such as exposure, policy duration, or other variables known when period $t$ is evaluated. When only a scalar exposure is required, we write it simply as $e_{it}$. The likelihood and all individual predictions are conditional on the history and period-level information actually available for that policyholder; equal history lengths are not required.

The multivariate response of interest is
\[
\bY_{it}=(Y_{it1},\ldots,Y_{itd})^{\trans}\in\mathcal Y.
\]
For example,
\[
\bY_{it}=(N_{it},L_{it},\mathbf S_{it}^{\trans})^{\trans},
\]
where $N_{it}$ is claim count, $L_{it}$ is aggregate loss, and
$\mathbf S_{it}=(S_{it1},\ldots,S_{itd_S})^{\trans}$
contains behavioural or operational measurements. The same construction can accommodate, for example, several lines of business, multiple correlated claim frequencies or severities, and combinations of insurance outcomes with auxiliary behavioural signals.

Let $K$ be the number of latent risk classes. The hierarchical model is
\begin{align}
\bp_i\mid\bx_i
&\sim\Dir(\phi\bpi_i),
\label{eq:dirichlet_model}\\
C_{it}\mid\bp_i
&\overset{\mathrm{ind}}{\sim}\Cat(\bp_i),
\qquad t=1,\ldots,T_i,
\label{eq:class_model}\\
\bY_{it}\mid C_{it}=k,\be_{it}
&\sim F_k(\cdot\mid\be_{it};\btheta_k),
\qquad k=1,\ldots,K.
\label{eq:expert_model}
\end{align}
Here $\phi$ is a concentration parameter, $\btheta_k$ are the expert distribution parameters, and $\bpi_i=(\pi_{i1},\ldots,\pi_{iK})^{\trans}$ is given by regression functions
\begin{equation}\label{eq:softmax}
\pi_{ik}
=
\frac{\exp(\bx_i^{\trans}\boldsymbol\beta_k)}{\sum_{\ell=1}^{K}\exp(\bx_i^{\trans}\boldsymbol\beta_\ell)},
\qquad k=1,\ldots,K,
\end{equation}
with $\boldsymbol\beta_K=\mathbf0$ for identification. Write
$\mathbf B=(\boldsymbol\beta_1,\ldots,\boldsymbol\beta_{K-1})$ for the free gating coefficients. The global parameters are
$\Theta=(\mathbf B,\phi,\btheta_1,\ldots,\btheta_K)$.

Each layer has a distinct actuarial interpretation.
\begin{itemize}[leftmargin=2em]
    \item $\bpi_i$ is the portfolio classification implied by observable attributes, that is, the risk information available from those attributes. It is also the prior mean of the individual risk composition.
    \item $\phi>0$ is the prior concentration, or total prior pseudo-count. A large value produces stronger shrinkage toward $\bpi_i$; a small value permits greater individual heterogeneity.
    \item $\bp_i$ is drawn once and remains fixed over time. It represents the policyholder's long-run risk composition rather than a dynamically evolving state.
    \item $C_{it}$ is the class expressed in period $t$. Different periods can express different classes even though $\bp_i$ is fixed.
    \item $F_k$ is the class-specific expert distribution for the complete multivariate outcome. Its parameters may incorporate exposure and other period-level information through $\be_{it}$.
\end{itemize}

The Dirichlet moments are
\begin{equation}\label{eq:p_moments}
\E(\bp_i\mid\bx_i)=\bpi_i,
\qquad
\Var(\bp_i\mid\bx_i)
=
\frac{\diag(\bpi_i)-\bpi_i\bpi_i^{\trans}}{\phi+1}.
\end{equation}
Thus, $\bpi_i$, $\phi$, and $\bp_i$ play different roles: the first is the portfolio centre, the second controls the strength of that information, and the third is the latent individual risk being learned.

The Dirichlet distribution in \cref{eq:dirichlet_model}, parameterised as $\Dir(\phi\bpi_i)$, uses a single total concentration parameter $\phi$ and imposes the usual Dirichlet covariance structure among the class proportions. This is a deliberate modelling trade-off rather than a claim that every portfolio follows an exact Dirichlet law. In return, the model obtains an interpretable prior centre, a scalar credibility strength, tractable updating conditional on the latent classes, and the exact predictive decomposition derived below. The restriction does not make the predictive model intrinsically inflexible: \Cref{subsec:theory} shows that rich expert families retain broad marginal approximation power and that the posterior predictive distribution can approach the true policyholder distribution under suitable conditions.

The expert specification should respect the support and structural relationships of the insurance outcomes. A practical and useful factorisation is a conditional class independence structure
\begin{equation}\label{eq:expert_factorization}
f_k(n,\ell,\mathbf s\mid\be)
=
f_{N,k}(n\mid\be)
 f_{L,k}(\ell\mid n,\be)
 \prod_{j=1}^{d_S} f_{S_j,k}(s_j\mid \be).
\end{equation}

Although conditional independence within a class may appear restrictive, it does not imply marginal independence among the outcomes. As \cref{subsec:theory} establishes, a sufficiently rich collection of support-appropriate experts and latent classes can represent complex multivariate dependence through the shared class structure. The factorisation also makes estimation substantially more tractable by allowing response-specific expert parameters to be updated separately.

\subsection{Quantities of interest}\label{subsec:quantities}

This subsection derives the main quantities used for actuarial modelling: the marginal distribution, moments and dependence, posterior risk classification, and, most importantly, the model's exact credibility representation for posterior prediction and premiums.

\subsubsection{Marginal distribution, moments, and dependence}

For each expert, define
\[
\bmu_{it,k}=\E(\bY_{it}\mid C_{it}=k,\be_{it}),
\qquad
\bSigma_{it,k}=\Var(\bY_{it}\mid C_{it}=k,\be_{it}).
\]
For multivariate cumulative distribution functions, inequalities such as
$\bY_{it}\leq\mathbf y$ are understood componentwise.

\begin{proposition}[One-period distribution and repeated-observation moments]\label{prop:marginal_moments}
For every policyholder $i$ and period $t$, the marginal cumulative distribution function is
\begin{equation}\label{eq:marginal_distribution}
F_{it}(\mathbf y\mid\bx_i,\be_{it})
=
\sum_{k=1}^{K}\pi_{ik}
F_k(\mathbf y\mid\be_{it};\btheta_k).
\end{equation}
If the first two expert moments exist, then
\begin{equation}\label{eq:marginal_mean}
\E(\bY_{it}\mid\bx_i,\be_{it})
=
\sum_{k=1}^{K}\pi_{ik}\bmu_{it,k},
\end{equation}
\begin{align}
\Var(\bY_{it}\mid\bx_i,\be_{it})
={}&
\sum_{k=1}^{K}\pi_{ik}
\left(
\bSigma_{it,k}+\bmu_{it,k}\bmu_{it,k}^{\trans}
\right)-
\left(\sum_{k=1}^{K}\pi_{ik}\bmu_{it,k}\right)
\left(\sum_{k=1}^{K}\pi_{ik}\bmu_{it,k}\right)^{\trans},
\label{eq:marginal_variance}
\end{align}
and, for $s\neq t$,
\begin{align}
\Cov(\bY_{is},\bY_{it}\mid\bx_i,\be_{is},\be_{it})
&=
\frac{1}{\phi+1}
\left[
\sum_{k=1}^{K}\pi_{ik}\bmu_{is,k}\bmu_{it,k}^{\trans}
-
\left(\sum_{k=1}^{K}\pi_{ik}\bmu_{is,k}\right)
\left(\sum_{k=1}^{K}\pi_{ik}\bmu_{it,k}\right)^{\trans}
\right].
\label{eq:cross_covariance}
\end{align}
\end{proposition}

\begin{proof}
Conditioning first on the period-specific class and then averaging over the individual composition gives
\begin{align*}
F_{it}(\mathbf y\mid\bx_i,\be_{it})
&=
\E\left[
\sum_{k=1}^{K}p_{ik}F_k(\mathbf y\mid\be_{it})
\middle|\bx_i
\right]\\
&=
\sum_{k=1}^{K}\pi_{ik}F_k(\mathbf y\mid\be_{it}),
\end{align*}
which proves \cref{eq:marginal_distribution}. The mean and variance follow from the standard finite-mixture formulas.

For $s\neq t$, the outcomes are conditionally independent given $\bp_i$, and
\[
\E(\bY_{it}\mid\bp_i,\be_{it})
=
\sum_k p_{ik}\bmu_{it,k}.
\]
The law of total covariance therefore reduces the cross-period covariance to the covariance of these two conditional means. Substituting
\[
\Cov(p_{ik},p_{i\ell}\mid\bx_i)
=
\frac{\pi_{ik}\mathbf1(k=\ell)-\pi_{ik}\pi_{i\ell}}{\phi+1}
\]
gives \cref{eq:cross_covariance}.
\end{proof}

The marginal distribution in \cref{eq:marginal_distribution} is a logit-weighted reduced mixture-of-experts (LRMoE) model in the baseline attributes $\bx_i$
\citep{FungBadescuLin2019Theory}. This connection is useful because it describes the flexibility of the one-period model as well as how to choose the experts' distributions; the main interpretation of the paper remains credibility through latent risk classes.

The concentration parameter also controls persistence. If
$Z_{itk}=\mathbf1\{C_{it}=k\}$, then, for $s\neq t$,
\[
\Corr(Z_{isk},Z_{itk}\mid\bx_i)=\frac{1}{\phi+1}.
\]
Smaller values of $\phi$ therefore produce stronger within-policyholder similarity in class expression.

\subsubsection{Posterior risk classification}

Let
\[
\cH_i=\sigma(\bx_i,\be_{i,1:T_i},\bY_{i,1:T_i})
\]
denote the observed history through the policyholder's most recent period $T_i$. For an assignment
$\mathbf c=(c_1,\ldots,c_{T_i})\in\{1,\ldots,K\}^{T_i}$, let
$n_k(\mathbf c)=\sum_{t=1}^{T_i}\mathbf1(c_t=k)$ be the number of periods assigned to class $k$, and define the multivariate beta function by
\[
\mathcal B(\mathbf a)=\frac{\prod_k\Gamma(a_k)}{\Gamma(\sum_k a_k)}.
\]

\begin{proposition}[Posterior composition and updated risk-class weights]\label{prop:posterior_weights}
Let
$g_{itk}=f_k(\bY_{it}\mid\be_{it};\btheta_k)$.
Assume these observed scores are finite and nonnegative and that the denominator in \cref{eq:assignment_weight} is strictly positive. The posterior density of $\bp_i$ is the finite mixture of Dirichlet distributions
\begin{equation}\label{eq:posterior_p_mixture}
h_i(\mathbf u\mid\cH_i)
=
\sum_{\mathbf c}
q_i(\mathbf c)
\,\Dir\{\mathbf u;\phi\bpi_i+\mathbf n(\mathbf c)\},
\end{equation}
where
\begin{equation}\label{eq:assignment_weight}
q_i(\mathbf c)
=
\frac{
\dfrac{\mathcal B\{\phi\bpi_i+\mathbf n(\mathbf c)\}}{\mathcal B(\phi\bpi_i)}
\prod_{t=1}^{T_i}g_{it,c_t}
}{
\sum_{\mathbf c'}
\dfrac{\mathcal B\{\phi\bpi_i+\mathbf n(\mathbf c')\}}{\mathcal B(\phi\bpi_i)}
\prod_{t=1}^{T_i}g_{it,c'_t}
}.
\end{equation}
If
\[
\tau_{itk}
=
\Prb(C_{it}=k\mid\cH_i)
=
\sum_{\mathbf c}q_i(\mathbf c)\mathbf1(c_t=k),
\]
then the posterior mean composition is
\begin{equation}\label{eq:posterior_mean_p}
\widehat p_{ik}
:=
\E(p_{ik}\mid\cH_i)
=
\frac{\phi\pi_{ik}+\sum_{t=1}^{T_i}\tau_{itk}}{\phi+T_i}.
\end{equation}
Equivalently,
\begin{equation}\label{eq:weight_credibility}
\widehat p_{ik}
=
\frac{\phi}{\phi+T_i}\pi_{ik}
+
\frac{T_i}{\phi+T_i}\bar\tau_{ik},
\qquad
\bar\tau_{ik}=\frac1{T_i}\sum_{t=1}^{T_i}\tau_{itk}.
\end{equation}
\end{proposition}

\begin{proof}
Conditional on $\mathbf C_i=\mathbf c$, Dirichlet--categorical conjugacy gives
\[
\bp_i\mid\mathbf C_i=\mathbf c,\cH_i
\sim
\Dir\{\phi\bpi_i+\mathbf n(\mathbf c)\}.
\]
Integrating the individual composition out of the joint density of the classes and outcomes gives the beta-function ratio in \cref{eq:assignment_weight}; normalisation then yields the assignment probabilities and the posterior mixture in \cref{eq:posterior_p_mixture}.

The mean of the Dirichlet component associated with $\mathbf c$ is
$(\phi\pi_{ik}+n_k(\mathbf c))/(\phi+T_i)$. Averaging over the assignment probabilities gives
\[
\E(p_{ik}\mid\cH_i)
=
\sum_{\mathbf c}q_i(\mathbf c)
\frac{\phi\pi_{ik}+n_k(\mathbf c)}{\phi+T_i}.
\]
Finally,
\[
\sum_{\mathbf c}q_i(\mathbf c)n_k(\mathbf c)
=
\sum_t\tau_{itk},
\]
which proves \cref{eq:posterior_mean_p}; rearranging gives \cref{eq:weight_credibility}.
\end{proof}

The vector
$\widehat\bp_i=(\widehat p_{i1},\ldots,\widehat p_{iK})^{\trans}$
is the insurer's updated risk classification for policyholder $i$. Equivalently, its entries are the posterior predictive probabilities of the classes expressed in a future period. It can be used directly for segmentation, monitoring, underwriting, and communication. A single label, such as
$\arg\max_k\widehat p_{ik}$, may be reported when an operational decision requires one class, but the full vector is more informative as a predictive class-probability summary. Similarly, the period-specific responsibilities $(\tau_{it1},\ldots,\tau_{itK})$ summarise the posterior class assignment for period $t$, whereas $\widehat\bp_i$ aggregates this information across the policyholder's history.

\Cref{eq:weight_credibility} makes the credibility structure explicit. The first term is the baseline portfolio classification obtained from observable characteristics, while the second summarises the policyholder's accumulated experience through posterior class responsibilities. The vector $\widehat\bp_i$ therefore provides a direct and interpretable bridge from traditional credibility weighting to multivariate risk classification.

\subsubsection{Posterior prediction and exact distributional credibility}

Consider period $T_i+1$ and suppose its expert predictors $\be_{i,T_i+1}$ are known. Define the portfolio-model and experience-updated cumulative distribution functions
\begin{align}
F_{i,T_i+1}^{\mathrm{pop}}(\mathbf y)
&=
\sum_{k=1}^{K}\pi_{ik}
F_k(\mathbf y\mid\be_{i,T_i+1};\btheta_k),
\label{eq:population_distribution}\\
F_{i,T_i+1}^{\mathrm{exp}}(\mathbf y)
&=
\sum_{k=1}^{K}\bar\tau_{ik}
F_k(\mathbf y\mid\be_{i,T_i+1};\btheta_k).
\label{eq:experience_distribution}
\end{align}
The first object is the predictive distribution obtained from the portfolio classification alone. The second is a posterior-responsibility-weighted fitted distribution that summarises the information in the policyholder's observed history; it is not the empirical distribution of the observed outcomes.

Let $\bmu_{i,T_i+1,k}$ and $\bSigma_{i,T_i+1,k}$ denote the mean and covariance of expert $k$ at $\be_{i,T_i+1}$.

\begin{proposition}[Exact model-based distributional credibility]\label{prop:predictive_credibility}
The posterior predictive cumulative distribution function is
\begin{equation}\label{eq:posterior_predictive}
F_{i,T_i+1}^{\mathrm{pred}}(\mathbf y\mid\cH_i)
=
\sum_{k=1}^{K}\widehat p_{ik}
F_k(\mathbf y\mid\be_{i,T_i+1};\btheta_k).
\end{equation}
It has the exact credibility representation
\begin{equation}\label{eq:distributional_credibility}
F_{i,T_i+1}^{\mathrm{pred}}(\mathbf y\mid\cH_i)
=
\frac{\phi}{\phi+T_i}F_{i,T_i+1}^{\mathrm{pop}}(\mathbf y)
+
\frac{T_i}{\phi+T_i}F_{i,T_i+1}^{\mathrm{exp}}(\mathbf y).
\end{equation}
The predictive mean is
\begin{equation}\label{eq:predictive_mean}
\E(\bY_{i,T_i+1}\mid\cH_i,\be_{i,T_i+1})
=
\sum_{k=1}^{K}\widehat p_{ik}\bmu_{i,T_i+1,k},
\end{equation}
and the predictive covariance for one future observation is
\begin{align}
\Var(\bY_{i,T_i+1}\mid\cH_i,\be_{i,T_i+1})
={}&
\sum_{k=1}^{K}\widehat p_{ik}
\left(
\bSigma_{i,T_i+1,k}
+\bmu_{i,T_i+1,k}\bmu_{i,T_i+1,k}^{\trans}
\right)
\notag\\
&-
\left(\sum_{k=1}^{K}\widehat p_{ik}\bmu_{i,T_i+1,k}\right)
\left(\sum_{k=1}^{K}\widehat p_{ik}\bmu_{i,T_i+1,k}\right)^{\trans}.
\label{eq:predictive_covariance}
\end{align}
\end{proposition}

\begin{proof}
Conditional on $\bp_i$, the next-period class has probabilities $p_{i1},\ldots,p_{iK}$. Therefore,
\begin{align*}
F_{i,T_i+1}^{\mathrm{pred}}(\mathbf y\mid\cH_i)
&=
\E\left[
\sum_{k=1}^{K}p_{ik}
F_k(\mathbf y\mid\be_{i,T_i+1})
\middle|\cH_i
\right]\\
&=
\sum_{k=1}^{K}\E(p_{ik}\mid\cH_i)
F_k(\mathbf y\mid\be_{i,T_i+1}),
\end{align*}
which proves \cref{eq:posterior_predictive}. Substituting the credibility form of $\widehat p_{ik}$ from \cref{eq:weight_credibility} and collecting terms gives \cref{eq:distributional_credibility}. The mean and covariance follow from the usual finite-mixture moment formulas with weights $\widehat p_{ik}$.
\end{proof}

\Cref{eq:distributional_credibility} is the central actuarial quantity and an exact model-based distributional credibility formula. Classical credibility is usually presented for a mean or premium. Here the same exact weighted-average structure holds for the entire future cumulative distribution function. This provides more information than a credibility formula for the mean alone: once the predictive distribution has been updated, the actuary obtains the predictive mean, quantiles, event probabilities, tail summaries, and other pricing or risk-management functionals from one coherent object. Bayesian nonparametric credibility methods also pursue full predictive inference \citep{SiuYang2009,HongMartin2017}, but the present finite-mixture construction retains a particularly simple closed-form decomposition. With sufficiently flexible expert families and a suitable number of classes, the finite and interpretable risk-class representation can approximate a broad range of multivariate insurance distributions.

The portfolio-model component stabilises predictions when individual experience is limited, while the experience-updated component personalises the fitted distribution. Both components are model-assisted distributions: they average the same support-appropriate experts, but use the baseline class probabilities and the posterior experience weights, respectively. This distinction makes the resulting quantities smoother and more informative than a raw empirical distribution based on a short policyholder history.

\subsubsection{Credibility premiums}

For pricing based on aggregate loss, marginalising \cref{eq:distributional_credibility} to $L$ preserves the exact credibility structure:
\begin{equation}\label{eq:aggregate_loss_cdf_credibility}
F_{L,i,T_i+1}^{\mathrm{pred}}(\ell\mid\cH_i)
=
\frac{\phi}{\phi+T_i}F_{L,i,T_i+1}^{\mathrm{pop}}(\ell)
+
\frac{T_i}{\phi+T_i}F_{L,i,T_i+1}^{\mathrm{exp}}(\ell).
\end{equation}
Thus, pricing can be based on the aggregate-loss margin while the class weights have still been learned from the complete multivariate history.

Let
$\mu^L_{i,T_i+1,k}=\E(L_{i,T_i+1}\mid C_{i,T_i+1}=k,\be_{i,T_i+1})$. The corresponding portfolio-model and experience-updated premiums are
\begin{align}
\nu_{i,T_i+1}^{\mathrm{pop}}
&=
\sum_{k=1}^{K}\pi_{ik}\mu^L_{i,T_i+1,k},
\label{eq:population_premium}\\
\nu_{i,T_i+1}^{\mathrm{exp}}
&=
\sum_{k=1}^{K}\bar\tau_{ik}\mu^L_{i,T_i+1,k}.
\label{eq:experience_premium}
\end{align}
Applying the mean functional to \cref{eq:aggregate_loss_cdf_credibility} gives the exact model-based premium credibility formula
\begin{align}
\widehat\nu_{i,T_i+1}
&=
\E(L_{i,T_i+1}\mid\cH_i,\be_{i,T_i+1})
\notag\\
&=
\frac{\phi}{\phi+T_i}\nu_{i,T_i+1}^{\mathrm{pop}}
+
\frac{T_i}{\phi+T_i}\nu_{i,T_i+1}^{\mathrm{exp}}.
\label{eq:net_premium}
\end{align}
Thus, the net premium is exactly a model-based credibility-weighted average of a portfolio-model premium and an experience-updated premium. Unlike the raw means that often appear in elementary credibility formulas, both terms are model-assisted: exposure, the multivariate outcome structure, and the fitted class-specific loss distributions are incorporated before the two sources of information are combined. An expected-value premium with loading $\eta_{\mathrm{prem}}$ is $(1+\eta_{\mathrm{prem}})\widehat\nu_{i,T_i+1}$. Other premium principles can likewise be applied directly to the predictive distribution of aggregate loss.

\subsection{Theoretical flexibility and predictive learning}\label{subsec:theory}

This subsection establishes two theoretical properties of the DMMM that are directly relevant to credibility. The first concerns \emph{portfolio flexibility}: with sufficiently rich experts and enough latent risk classes, the one-period model can approximate a broad range of multivariate regression distributions. The second concerns \emph{individual learning}: once a collection of risk classes is fixed, an increasing policyholder history learns the combination of those classes that best represents that policyholder. \Cref{sec:dmm-simulation-study} illustrates both properties numerically.

Briefly, the theoretical properties follow from two standard mixture-model results. First, the one-period marginal distribution in \cref{eq:marginal_distribution} is an LRMoE, so existing denseness results apply directly. Second, conditional on a fixed dictionary of experts, learning the policyholder-specific composition is a finite-dimensional Bayesian mixture problem, for which posterior concentration around the Kullback--Leibler (KL) best approximation is well known. We now state the formal results.

\begin{theorem}[LRMoE; \citet{FungBadescuLin2019Theory}]\label{thm:lrmoe_dense}
Fix the period-level information $\be$ and let $\mathcal X$ be a compact
rectangular covariate space. Suppose that the family
$\{F^0(\cdot\mid\bx,\be):\bx\in\mathcal X\}$ is tight and that
$F^0(\by\mid\bx,\be)$ is Lipschitz continuous in $\bx$ for every
continuity point $\by$. If the expert families satisfy the denseness
conditions of \citet{FungBadescuLin2019Theory}, then there exists a
sequence of finite LRMoE models of the form in \cref{eq:marginal_distribution} such that
\[
\sup_{\bx\in\mathcal X}
\left|
F_K(\by\mid\bx,\be)-F^0(\by\mid\bx,\be)
\right|
\longrightarrow 0,
\qquad K\to\infty,
\]
for every continuity point $\by$ of the target distribution.
\end{theorem}

Thus, although the multinomial-logit gate uses affine predictors in the baseline covariates, a sufficiently rich collection of latent classes can approximate nonlinear regression effects, multimodality, and complex dependence structures. The result concerns approximation capacity; increasing $K$ alone is insufficient unless the chosen expert families satisfy the required denseness conditions. \citet{FungBadescuLin2019Theory} discusses several expert families with the required approximation properties, including Gamma experts for positive continuous outcomes and Gamma-count experts for count outcomes.

For predictive learning, fix a policyholder and suppress $\bx$ and $\be$. For a $K$-expert dictionary, define
\[
f_{\bp,K}=\sum_{k=1}^K p_k f_{k,K},
\qquad
\delta_K=\inf_{\bp\in\simplex}\KL(f^0,f_{\bp,K}),
\]
where $f^0$ is the policyholder's true distribution. We call the dictionaries \emph{KL-rich} if
\begin{equation}\label{eq:KL_rich}
\delta_K\longrightarrow0\qquad(K\to\infty).
\end{equation}
This is stronger than weak denseness, but can be ensured by suitable support-consistent expert families under standard regularity conditions. This is a condition on how the modeller enlarges the dictionary, rather than an assumption that the true portfolio literally consists of an increasing number of discrete risk types.

\begin{theorem}[Posterior predictive learning]\label{thm:predictive_truth}
Suppose $\bY_1,\bY_2,\ldots$ are i.i.d.\ from $f^0$, conditional on the global model parameters. For each fixed $K$, assume that
$R_K(\bp)=\KL(f^0,f_{\bp,K})$ is finite and continuous on the simplex and has a unique minimiser $\bp_K^\star$, that the usual uniform law of large numbers holds for the log likelihood, and that the Dirichlet prior has strictly positive parameters. Then
\[
\KL\!\left(f^0,f_{T,K}^{\mathrm{pred}}\right)
\xrightarrow[T\to\infty]{\Prb_0}
\delta_K,
\qquad
f_{T,K}^{\mathrm{pred}}
=
\sum_{k=1}^K
\E(p_k\mid\bY_{1:T})f_{k,K}.
\]
Consequently, if \cref{eq:KL_rich} holds,
\begin{equation}\label{eq:predictive_truth_limit}
\lim_{K\to\infty}\lim_{T\to\infty}
\KL\!\left(f^0,f_{T,K}^{\mathrm{pred}}\right)=0
\qquad\text{in }\Prb_0\text{-probability}.
\end{equation}
\end{theorem}

\begin{proof}
For fixed $K$, standard misspecified Bayesian consistency
\citep{Berk1966,KleijnVanDerVaart2006} gives posterior concentration at
$\bp_K^\star$. Hence the posterior mean composition converges to
$\bp_K^\star$, and continuity of $R_K$ gives the first limit. The second follows from $\delta_K\to0$.
\end{proof}

\begin{corollary}[Premium learning]\label{cor:moment_learning}
If the loss coordinate additionally has $(1+\eta)$ moments that are uniformly bounded, for some $\eta>0$, under $f^0$ and across the expert dictionaries, then the corresponding posterior predictive expected-loss premium converges to the true expected loss under the same sequential limit.
\end{corollary}

The interpretation is straightforward. For a fixed set of risk classes, a longer policyholder history learns the best combination of those classes. Increasing $K$ makes the collection of available risk classes richer, and under the KL-richness condition the remaining approximation error can be made arbitrarily small. The result is therefore a controlled finite-mixture analogue of standard Bayesian mixture consistency. It describes the capacity of the model, not a recommendation to choose a very large $K$ in finite samples; in practice, the number of classes should still be controlled through actuarial judgment and out-of-sample validation.

\section{Estimation and inference}\label{sec:estimation}

This section develops an expectation--maximization (EM)-type estimation procedure for the DMMM. Standard finite-mixture models are commonly fitted by EM; here we additionally allow fixed response-component weights so that auxiliary signals can inform class formation without dominating it simply because they are numerous. The resulting procedure maximises a weighted criterion while preserving the Dirichlet updating structure. We first define the weighted criterion and the corresponding E- and M-step quantities, and then discuss the selection of the number of classes and response weights, as well as bootstrap inference.

A Julia implementation of the DMMM model and the estimation procedure developed in this section is available at \url{https://github.com/ianwengchan/DMMM.jl}.

\subsection{Weighted EM-type estimation}\label{subsec:em}

Suppose that, conditional on class $k$, the expert admits the factorisation
\begin{equation}\label{eq:weighted_expert_factorization}
f_k(\bY_{it}\mid\be_{it};\btheta_k)
=
\prod_{r=1}^{R}
f_{rk}\!\left(
Y_{it}^{(r)}
\mid \bY_{it}^{(\mathcal P_r)},\be_{it};
\btheta_{rk}
\right),
\end{equation}
where $\mathcal P_r$ indexes any response components on which factor $r$ is conditionally modelled; it may be empty. This notation includes, for example, the factorisation
$f_{N,k}(n\mid\be)f_{L,k}(\ell\mid n,\be)\prod_j f_{S_j,k}(s_j\mid\be)$
used in \cref{eq:expert_factorization}. For notational simplicity, the additional response-conditioning arguments are suppressed below.

Let $\omega_r>0$ be the fixed importance assigned to response component $r$, and define the weighted class score
\begin{equation}\label{eq:weighted_class_score}
\widetilde g_{itk}
=
\prod_{r=1}^{R}
\left[
f_{rk}\!\left(Y_{it}^{(r)}\mid\be_{it};\btheta_{rk}\right)
\right]^{\omega_r},
\qquad
\log\widetilde g_{itk}
=
\sum_{r=1}^{R}\omega_r
\log f_{rk}\!\left(Y_{it}^{(r)}\mid\be_{it};\btheta_{rk}\right).
\end{equation}

The purpose of \cref{eq:weighted_class_score} is not to claim that one variable is intrinsically more informative than another. Rather, it allows the estimated classes to reflect the actuarial target of the analysis. For example, suppose the response consists of claim count $N_{it}$, aggregate loss $L_{it}$, and $d_S$ auxiliary signals $S_{it1},\ldots,S_{itd_S}$. If the classes are intended primarily to represent insurance risk, one possible specification is
\begin{equation}\label{eq:loss_oriented_weights}
\omega_N=\omega_L=1,
\qquad
\omega_{S_j}=\frac{\omega_{\mathrm{sig}}}{d_S},
\quad j=1,\ldots,d_S,
\qquad
0<\omega_{\mathrm{sig}}\leq1.
\end{equation}
The two insurance outcomes therefore retain unit weight, whereas the \emph{entire} signal block contributes total weight $\omega_{\mathrm{sig}}$. Setting $\omega_{\mathrm{sig}}=1$ gives the signal block the same total weight as each insurance outcome, regardless of the number of signals. This provides a simple default when the signals are intended to provide collateral information without dominating class formation. More generally, $\omega_{\mathrm{sig}}$ can be treated as a tuning parameter and selected by policyholder-level validation for the actuarial prediction target. The ordinary unweighted likelihood is recovered only when every individual component weight $\omega_r$ equals one.

\subsubsection{Complete-data weighted log-likelihood}\label{subsec:likelihoods}

Introduce the indicators $Z_{itk}=\mathbf1\{C_{it}=k\}$. Following the usual mixture-model argument, the corresponding complete-data weighted criterion is
\begin{align}
\widetilde\ell_{c,i}(\Theta;\boldsymbol\omega)
={}&
\log\Gamma(\phi)
-
\sum_{k=1}^{K}\log\Gamma(\phi\pi_{ik})
+
\sum_{k=1}^{K}(\phi\pi_{ik}-1)\log p_{ik}
\notag\\
&+
\sum_{t=1}^{T_i}\sum_{k=1}^{K}
Z_{itk}
\left\{
\log p_{ik}
+
\sum_{r=1}^{R}\omega_r
\log f_{rk}\!\left(Y_{it}^{(r)}\mid\be_{it};\btheta_{rk}\right)
\right\}.
\label{eq:complete_likelihood}
\end{align}
The Dirichlet distribution and categorical class-assignment terms are unchanged from the ordinary complete-data log-likelihood; only the response contribution is reweighted. This distinction preserves the interpretation of $\phi\bpi_i$ as prior pseudo-counts while allowing the observed outcomes to have different degrees of influence on the inferred classes.

\subsubsection{Weighted E-step}\label{subsec:estep}

At iteration $m$, let $\Theta^{(m)}$ denote the current parameter values. For a complete class assignment $\mathbf c$, define
\begin{equation}\label{eq:EM_assignment_weight}
\gamma_i^{(m)}(\mathbf c)
=
\frac{
\dfrac{\mathcal B\{\phi^{(m)}\bpi_i^{(m)}+\mathbf n(\mathbf c)\}}
{\mathcal B(\phi^{(m)}\bpi_i^{(m)})}
\prod_{t=1}^{T_i}\widetilde g_{it,c_t}^{(m)}
}{
\sum_{\mathbf c'}
\left(
\dfrac{\mathcal B\{\phi^{(m)}\bpi_i^{(m)}+\mathbf n(\mathbf c')\}}
{\mathcal B(\phi^{(m)}\bpi_i^{(m)})}
\prod_{t=1}^{T_i}\widetilde g_{it,c'_t}^{(m)}
\right)
}.
\end{equation}
For the ordinary model, these are the posterior probabilities of the complete class sequence in \cref{eq:assignment_weight}. Under non-unit response weights, they are better interpreted as \emph{generalized-posterior} assignment probabilities obtained by normalising the powered kernel over the latent variables. In either case, they sum to one over all class sequences and play the same computational role in the E-step.

\begin{proposition}[Weighted E-step expectations]\label{prop:EM_estep}
The quantities required by the weighted EM objective are
\begin{equation}\label{eq:tau_EM}
\widetilde\tau_{itk}^{(m)}
:=
\sum_{\mathbf c}
\gamma_i^{(m)}(\mathbf c)\mathbf1(c_t=k),
\end{equation}
and
\begin{equation}\label{eq:xi_EM}
\widetilde\xi_{ik}^{(m)}
:=
\sum_{\mathbf c}\gamma_i^{(m)}(\mathbf c)
\left[
\psi\{\phi^{(m)}\pi_{ik}^{(m)}+n_k(\mathbf c)\}
-
\psi\{\phi^{(m)}+T_i\}
\right],
\end{equation}
where $\psi$ is the digamma function. Up to terms that do not involve the candidate parameters, the weighted $Q$-function is
\begin{align}
\widetilde Q(\Theta\mid\Theta^{(m)};\boldsymbol\omega)
={}&
\sum_{i=1}^{n}
\left[
\log\Gamma(\phi)
-
\sum_{k=1}^{K}\log\Gamma(\phi\pi_{ik})
+
\sum_{k=1}^{K}
(\phi\pi_{ik}-1)\widetilde\xi_{ik}^{(m)}
\right]
\notag\\
&+
\sum_{i=1}^{n}\sum_{t=1}^{T_i}\sum_{k=1}^{K}
\widetilde\tau_{itk}^{(m)}
\sum_{r=1}^{R}\omega_r
\log f_{rk}\!\left(Y_{it}^{(r)}\mid\be_{it};\btheta_{rk}\right).
\label{eq:Q_EM}
\end{align}
\end{proposition}

\begin{proof}
Conditional on $\mathbf C_i=\mathbf c$, the response weights do not involve $\bp_i$. Hence the Dirichlet--categorical conjugacy is unchanged and
\[
\mathcal L_{\boldsymbol\omega}
(\bp_i\mid\mathbf C_i=\mathbf c,\bY_i,\Theta^{(m)})
=
\Dir\{\phi^{(m)}\bpi_i^{(m)}+\mathbf n(\mathbf c)\},
\]
where $\mathcal L_{\boldsymbol\omega}$ denotes the normalised generalized-posterior law. The usual Dirichlet log-moment therefore gives
\[
\E_{\boldsymbol\omega}
(\log p_{ik}\mid\mathbf C_i=\mathbf c,\bY_i,\Theta^{(m)})
=
\psi\{\phi^{(m)}\pi_{ik}^{(m)}+n_k(\mathbf c)\}
-
\psi\{\phi^{(m)}+T_i\}.
\]
Averaging the log-moment and the indicators $Z_{itk}$ over \cref{eq:EM_assignment_weight} gives \cref{eq:tau_EM,eq:xi_EM}. Taking the same weighted conditional expectation of \cref{eq:complete_likelihood} yields \cref{eq:Q_EM}. As in the ordinary EM derivation, the categorical term $\sum_{t,k}Z_{itk}\log p_{ik}$ does not contain candidate values of $\mathbf B$, $\phi$, or the expert parameters after the expectation has been taken, and can be omitted from the maximisation criterion.
\end{proof}

Direct evaluation of these quantities from \cref{eq:EM_assignment_weight} sums over $K^{T_i}$ class sequences and can therefore be computationally expensive. Computational details are provided in the Appendix. \Cref{subsec:computing_estep} develops two exact alternatives based on grouping and count-state dynamic programming, together with a collapsed Monte Carlo approximation for larger problems.

\subsubsection{Weighted M-step}\label{subsec:mstep}

\begin{proposition}[Separation of the weighted M-step]\label{prop:EM_mstep}
Assume the expert parameter blocks are variation-independent and disjoint across factors and classes. When a finite conditional maximiser exists, the parameters of factor $r$ in class $k$ are updated through
\begin{equation}\label{eq:expert_update}
\btheta_{rk}^{(m+1)}
=
\arg\max_{\btheta_{rk}}
\sum_{i=1}^{n}\sum_{t=1}^{T_i}
\widetilde\tau_{itk}^{(m)}\omega_r
\log f_{rk}\!\left(Y_{it}^{(r)}\mid\be_{it};\btheta_{rk}\right).
\end{equation}
The gating coefficients and concentration parameter maximise
\begin{align}
\widetilde Q_D(\mathbf B,\phi)
=
\sum_{i=1}^{n}
\Bigg[
&\log\Gamma(\phi)
-
\sum_{k=1}^{K}\log\Gamma\{\phi\pi_{ik}(\mathbf B)\}
+
\sum_{k=1}^{K}
\{\phi\pi_{ik}(\mathbf B)-1\}\widetilde\xi_{ik}^{(m)}
\Bigg].
\label{eq:QD_EM}
\end{align}
\end{proposition}

\begin{proof}
The expert parameters occur only in the last line of \cref{eq:Q_EM}. Under the stated disjoint-block assumption, parameters belonging to different factors and classes separate, yielding \cref{eq:expert_update}. Factorisation alone would not suffice if factors shared parameters or constraints. The remaining candidate-dependent terms are exactly the expected log-density of the Dirichlet layer and yield \cref{eq:QD_EM}.
\end{proof}

\paragraph{Updating the experts.}
Obtaining the maximisers in \cref{eq:expert_update} is analogous to standard weighted likelihood estimation for the chosen expert distribution, which is usually a straightforward procedure. The substantive effect of the weighting enters through the responsibilities $\widetilde\tau_{itk}^{(m)}$, which determine which observations are attributed to each class. Thus, once the weighted E-step has formed the classes, the individual expert updates remain ordinary weighted distribution fits with $\widetilde\tau_{itk}^{(m)}$ acting as fractional class labels.

When an expert fit is ill-conditioned or a class has few effective observations, standard component-specific penalties or regularisation can be added to \cref{eq:expert_update} without changing the E-step.

\paragraph{Updating the risk-class regression.}
The influence of the weighting comes through the E-step quantities $\widetilde\xi_{ik}^{(m)}$. Conditional on these, the criterion $\widetilde Q_D$ has the same functional form as a mean--precision Dirichlet-regression log-likelihood \citep{HijaziJernigan2009,Maier2014}, with the complete-data log-compositions replaced by their conditional expectations. We therefore update $\mathbf B$ and $\phi$ by standard numerical Dirichlet-regression optimisation of \cref{eq:QD_EM}, imposing $\boldsymbol\beta_K=\mathbf0$ for identification and parameterising $\phi=\exp(\rho)$ to enforce positivity. Existing implementations such as \texttt{DirichletReg} \citep{Maier2014} provide useful computational templates, while the full hierarchical model can alternatively be fitted in a Bayesian framework using, for example, Stan \citep{CarpenterEtAl2017} or PyMC \citep{AbrilPlaEtAl2023}.

\subsubsection{Selecting the number of classes and response weights}\label{subsec:selecting_K}

The number of classes $K$ and, when applicable, the response weights are tuning choices rather than EM outputs. We select them by policyholder-level validation, keeping all observations from the same policyholder in the same data split. In practice, we fit a moderate grid of candidate values from several starting points and compare held-out predictive performance for the actuarial target of interest, together with class stability, effective class sizes, and interpretability. We prefer the smallest $K$ and least aggressive weighting whose out-of-sample predictive performance is practically comparable with the best candidate. Multiple starting values remain advisable because mixture criteria can contain local maxima. Once the number of classes and weights are selected, the model may be refitted on the combined training and validation samples before final evaluation on an untouched test portfolio.

\subsection{Statistical inference}\label{subsec:inference}

The EM/ECM algorithm produces point estimates but not automatic standard errors. A nonparametric policyholder cluster bootstrap can be used by resampling complete policyholder records and refitting the same weighted procedure. This preserves within-policy dependence and does not require interpreting the weighted criterion as a generative likelihood.

For each bootstrap sample $b$, let $\widehat\Theta^{*(b)}$ denote the refitted parameter vector. Standard errors and confidence intervals for parameters or smooth predictive functionals are obtained from the empirical distribution of the corresponding bootstrap estimates.

Conditional on fitted parameters, prediction intervals for a future scalar response are obtained directly from quantiles of
$F_{i,T_i+1}^{\mathrm{pred}}$. To incorporate parameter uncertainty, recompute the posterior weights and predictive distribution after each bootstrap refit,
\[
F_{i,T_i+1}^{\mathrm{pred},*(b)}(y)
=
\sum_{k=1}^{K}
\widehat p_{ik}^{*(b)}
F_k\!\left(
y\mid\be_{i,T_i+1};
\widehat\btheta_k^{*(b)}
\right).
\]
Sampling from these replicate-specific predictive distributions yields a bootstrap predictive distribution for a future outcome. If the quantity of interest is instead a premium or another predictive functional, evaluate that functional in every bootstrap replicate and form its interval from the resulting empirical distribution.

\section{Simulation study}
\label{sec:dmm-simulation-study}

This section uses a deliberately challenging simulation design to illustrate the two theoretical properties discussed in \cref{subsec:theory}: approximation of a complex portfolio distribution and policyholder-level learning from repeated experience. We model claim frequency and aggregate losses across two business lines and additionally examine the claim-severity distribution implied by the fitted model. The insurer wishes to perform experience rating for a joint policy accounting for both business lines.

\subsection{Insurance setting and data-generating mechanism}

Consider a policy carrying two coverages, such as personal motor own-damage or collision coverage and third-party liability coverage. For Line~1, both annual claim count $N_1$ and aggregate loss $L_1$ are retained. For Line~2, the modelling input contains only annual aggregate loss $L_2$. The fitted response is therefore
\[
\bY=(N_1,L_1,L_2)^\trans,
\qquad
L_\ell=\sum_{j=1}^{N_\ell}X_{\ell j}.
\]

We simulate a portfolio of $10{,}000$ policyholders. Each policyholder has a persistent latent risk $U\sim\mathcal N(0,1)$. Conditional on $U=u$, annual counts follow zero-inflated negative-binomial distributions and individual severities follow two-component lognormal mixtures. Specifically,
\[
N_{\ell}\mid U=u
\sim
\varrho_\ell(u)\delta_0+
\{1-\varrho_\ell(u)\}\operatorname{NB}_2\{\lambda_\ell(u),\kappa_\ell(u)\},
\]
with
\begin{align*}
\lambda_1(u)&=0.22e^{0.90u-0.90^2/2},
& \varrho_1(u)&=\operatorname{logit}^{-1}(0.20-1.50u),
& \kappa_1(u)&=1.20e^{0.35u},\\
\lambda_2(u)&=0.11e^{0.95u-0.95^2/2},
& \varrho_2(u)&=\operatorname{logit}^{-1}(0.35-1.35u),
& \kappa_2(u)&=e^{0.30u}.
\end{align*}
Writing $\operatorname{LN}(m,c)$ for a lognormal distribution with mean $m$ and coefficient of variation $c$,
\begin{align*}
X_1\mid U=u
&\sim
e^{0.08u}\left[
\{1-\chi_1(u)\}\operatorname{LN}(1800,0.26)
+\chi_1(u)\operatorname{LN}(8500,0.20)
\right],\\
X_2\mid U=u
&\sim
e^{0.10u}\left[
\{1-\chi_2(u)\}\operatorname{LN}(3200,0.30)
+\chi_2(u)\operatorname{LN}(13500,0.17)
\right],
\end{align*}
where
\[
\chi_1(u)=\operatorname{logit}^{-1}\{\operatorname{logit}(0.25)+2.20u\},
\qquad
\chi_2(u)=\operatorname{logit}^{-1}\{\operatorname{logit}(0.22)+2.00u\}.
\]
Thus, latent risk affects both the scale and the modality of the severity distributions.

The parameter values are chosen to produce plausible insurance-portfolio magnitudes while deliberately accentuating features that are difficult for simple parametric models, including zero inflation, overdispersion, bimodal severity, and dependence across business lines. This design makes the approximation problem visible and tests whether a DMMM with simpler component distributions can recover the resulting portfolio and policyholder-level distributions.

\subsection{DMMM specification}

For latent class $k$, we use the factorisation
\[
f_k(n,\ell_1,\ell_2)
=
\operatorname{Pois}(n;\lambda_k)\,
\operatorname{Gamma}(\ell_1;n\eta_k,\vartheta_k)\,
\operatorname{Tweedie}(\ell_2;\mu_k,\varphi_k^{\mathrm{Tw}},\nu_k^{\mathrm{Tw}}),
\]
where $\eta_k$ and $\vartheta_k$ are the Gamma shape and scale parameters, respectively; the aggregate-loss term is understood as a point mass at zero when $n=0$; and $\varphi_k^{\mathrm{Tw}}$ and $\nu_k^{\mathrm{Tw}}$ are the Tweedie dispersion and power parameters, with $1<\nu_k^{\mathrm{Tw}}<2$. For both lines, this specification has a compound-Poisson--Gamma interpretation within each expert. In particular, for Line~1,
\[
N_1\mid k\sim\operatorname{Pois}(\lambda_k),
\qquad
X_{1j}\mid k\sim\operatorname{Gamma}(\eta_k,\vartheta_k),
\qquad
L_1=\sum_{j=1}^{N_1}X_{1j}.
\]

Although $X_1$ is not itself an element of the fitted response $\bY$, the DMMM therefore implies a distribution for realised claim severity. For portfolio class weights
$\bpi=(\pi_1,\ldots,\pi_K)^\trans$,
\[
f_{X_1}(x;\bpi)
=
\sum_{k=1}^K \pi_k^{(X)}
\operatorname{Gamma}(x;\eta_k,\vartheta_k),
\qquad
\pi_k^{(X)}
=
\frac{\pi_k\lambda_k}{\sum_{h=1}^K\pi_h\lambda_h}.
\]
The frequency adjustment is required because a randomly selected claim is more likely to arise from an expert with a larger Poisson rate; this is the usual size-biasing induced when claim frequency and severity are dependent
\citep{BlierWong2026}. The Tweedie component has the analogous compound-Poisson--Gamma representation for Line~2, but the illustration will focus on Line~1 for simplicity.

In an empirical application, the expert parameters would ordinarily be estimated jointly with the class structure. Here, to isolate how approximation quality changes with the richness of the expert dictionary, we fix a nested sequence of expert sets. A master grid of 100 experts is specified before estimation using five risk anchors for each line and two severity regimes per line. The candidate expert sets for $K\in\{4,12,20,36,68,100\}$ are literal nested subsets of this grid. For each value of $K$, only the portfolio weights are estimated by maximum likelihood from the simulated sample of $10{,}000$ policyholders.

\subsection{Portfolio approximation}

\Cref{tab:dmm-sim-portfolio} reports binned Kullback--Leibler (KL) divergence for the four displayed quantities and the joint response $(N_1,L_1,L_2)$ as the number of experts increases. \Cref{fig:dmm-sim-portfolio} compares the corresponding fitted probability mass functions and densities. Approximation quality improves sharply as the expert dictionary becomes richer, with most of the improvement occurring by $K=12$--$20$. The $K=36$ model provides a further, although smaller, reduction in joint KL divergence while remaining relatively parsimonious. Importantly, the implied distribution of $X_1$ also closely recovers the true claim-severity distribution even though individual severity is not included directly in the fitted response.

\begin{table}[htbp]
\centering
\caption{Portfolio approximation on the $10{,}000$-policyholder sample (KL divergence, nats).}
\label{tab:dmm-sim-portfolio}
\small
\begin{tabular}{rrrrrr}
\toprule
$K$ & $N_1$ & $X_1$ & $L_1$ & $L_2$ & Joint\\
\midrule
4 & 0.1022 & 0.1947 & 0.0992 & 0.0481 & 0.1391 \\
12 & 0.0024 & 0.0340 & 0.0087 & 0.0053 & 0.0155 \\
20 & 0.0041 & 0.0260 & 0.0072 & 0.0047 & 0.0104 \\
36 & 0.0034 & 0.0263 & 0.0072 & 0.0045 & 0.0078 \\
68 & 0.0031 & 0.0264 & 0.0069 & 0.0044 & 0.0070 \\
100 & 0.0031 & 0.0264 & 0.0069 & 0.0045 & 0.0072 \\
\bottomrule
\end{tabular}
\end{table}

The shared expert membership also reproduces much of the dependence structure. \Cref{tab:dmm-sim-corr} reports the correlations implied by the nested expert sets. The fitted DMMM captures the strong association between frequency and aggregate loss, the weaker cross-line dependence, and the associations involving individual severity. The latter associations are notable because $X_1$ is recovered only through the compound representation of the Line~1 expert.

\begin{table}[htbp]
\centering
\caption{True and implied correlations across the nested expert sets.}
\label{tab:dmm-sim-corr}
\small
\begin{tabular}{lrrrrrr}
\toprule
 & $\rho(N_1,L_1)$ & $\rho(N_1,L_2)$ & $\rho(L_1,L_2)$ & $\rho(N_1,X_1)$ & $\rho(L_1,X_1)$ & $\rho(L_2,X_1)$\\
\midrule
True & 0.935 & 0.234 & 0.271 & 0.286 & 0.501 & 0.210 \\
$K=4$ & 0.886 & 0.000 & 0.003 & 0.000 & 0.881 & 0.030 \\
$K=12$ & 0.931 & 0.361 & 0.377 & 0.299 & 0.617 & 0.227 \\
$K=20$ & 0.920 & 0.283 & 0.307 & 0.261 & 0.619 & 0.212 \\
$K=36$ & 0.918 & 0.218 & 0.240 & 0.248 & 0.619 & 0.192 \\
$K=68$ & 0.918 & 0.197 & 0.222 & 0.241 & 0.618 & 0.192 \\
$K=100$ & 0.917 & 0.197 & 0.222 & 0.238 & 0.617 & 0.192 \\
\bottomrule
\end{tabular}
\end{table}

\begin{figure}[H]
\centering
\includegraphics[width=.49\textwidth]{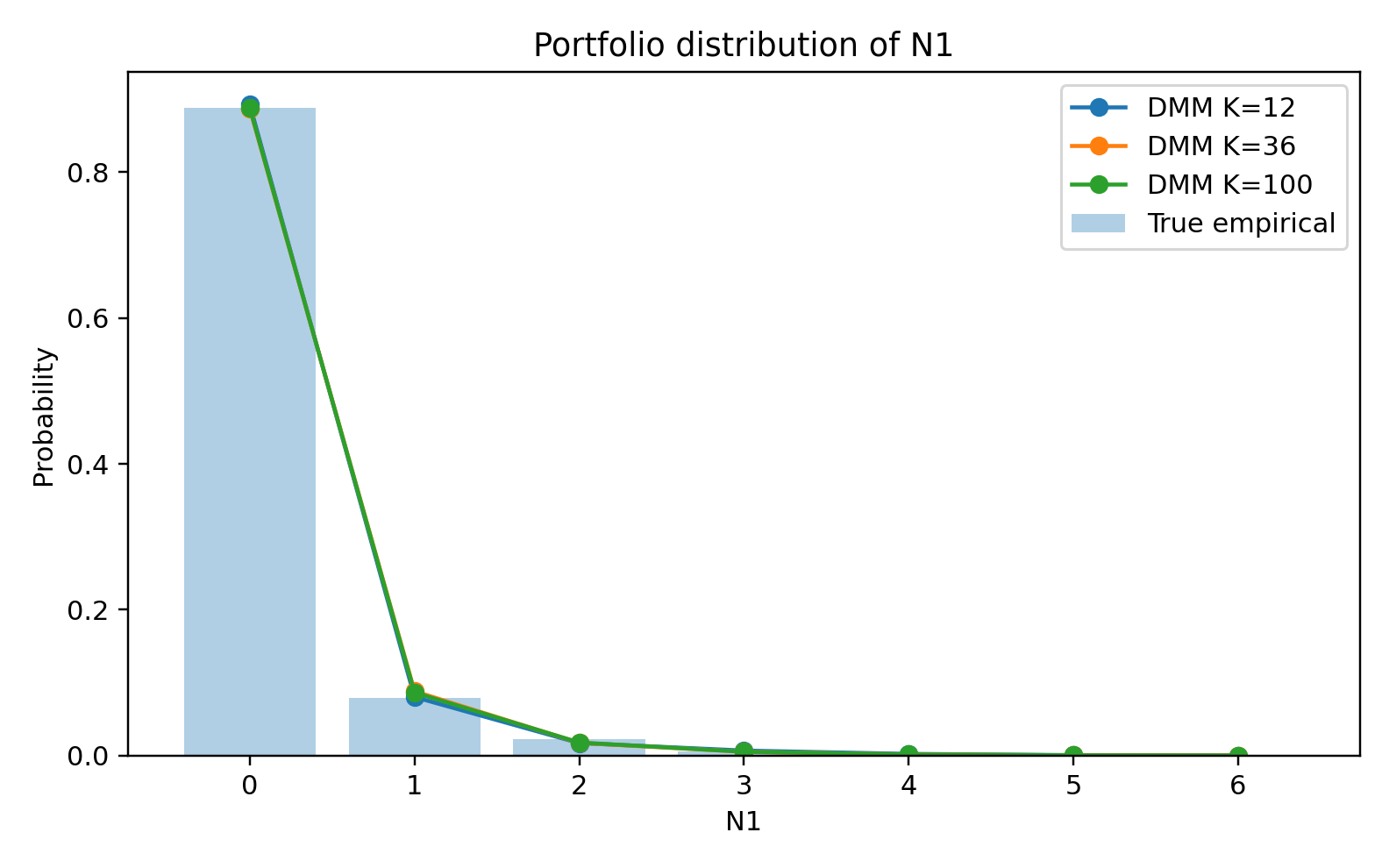}\hfill
\includegraphics[width=.49\textwidth]{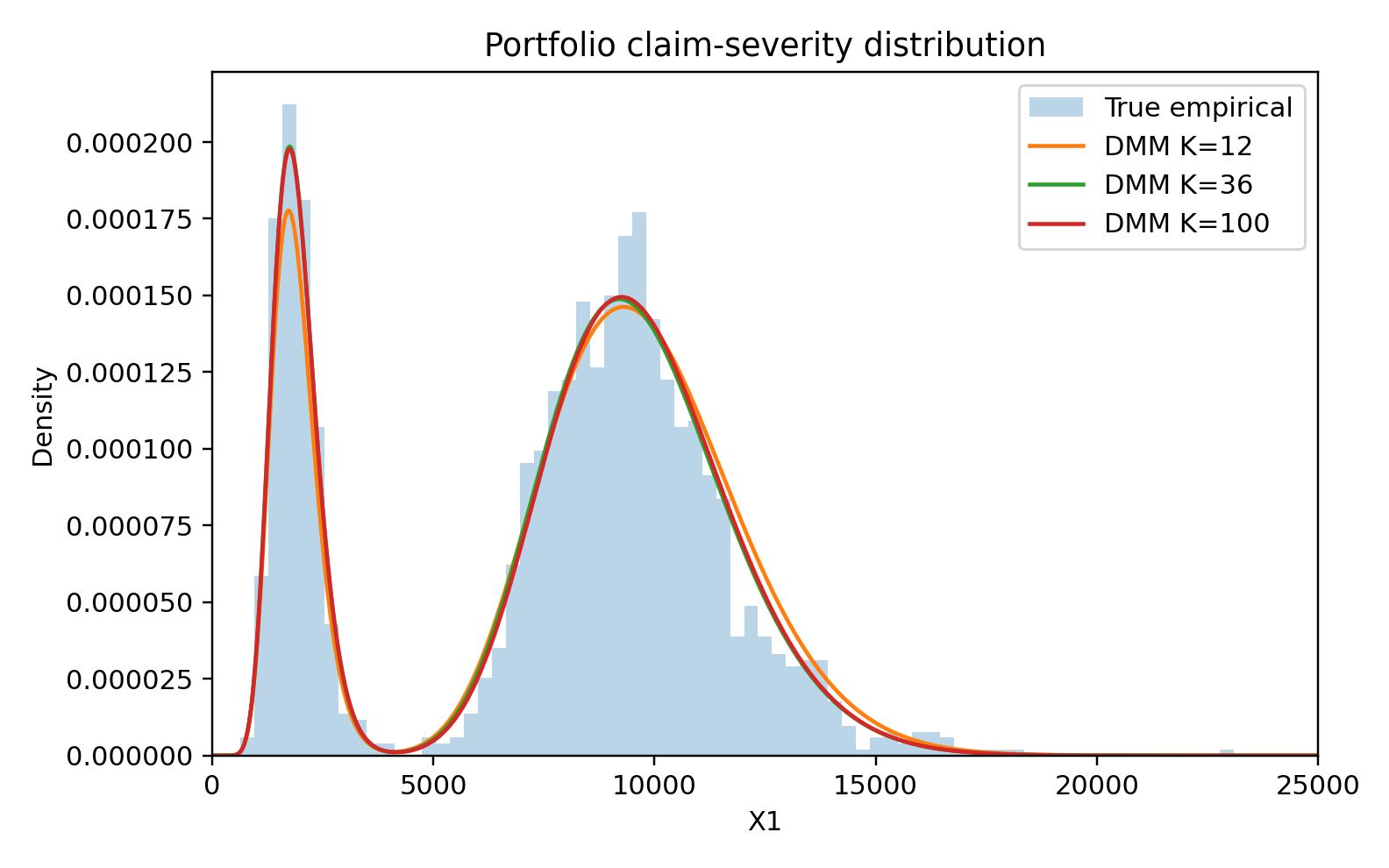}\\[0.4em]
\includegraphics[width=.49\textwidth]{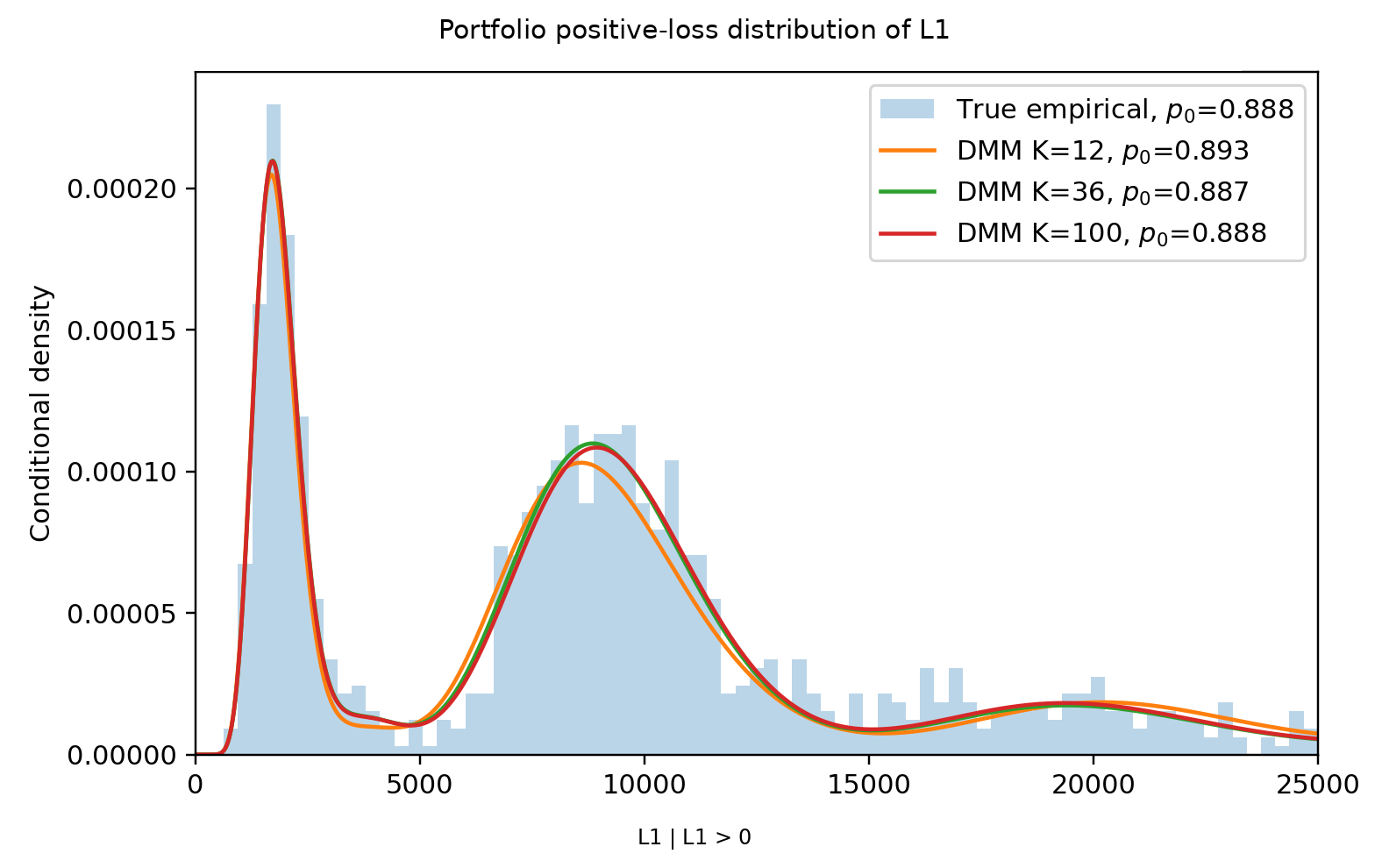}\hfill
\includegraphics[width=.49\textwidth]{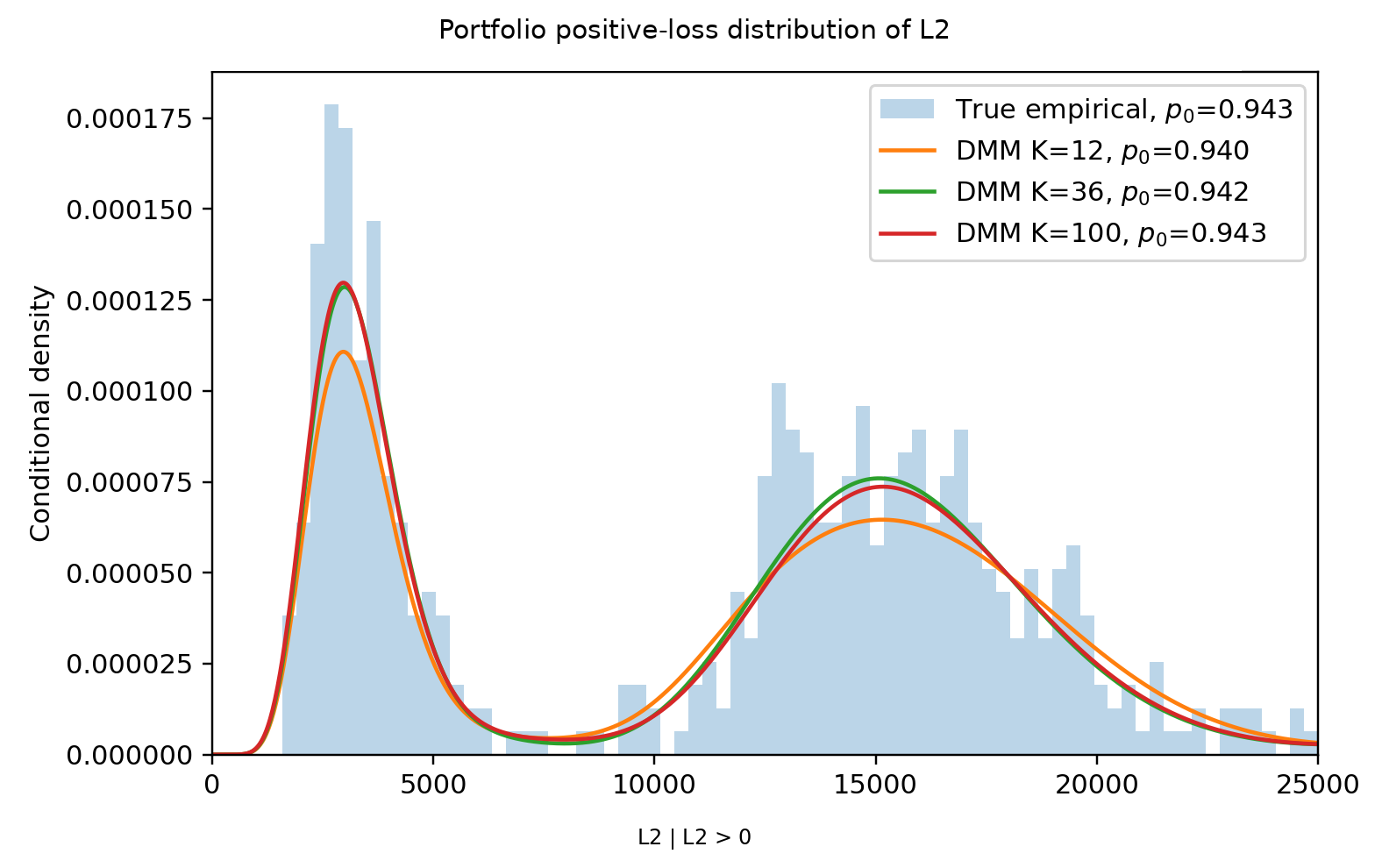}
\caption{Portfolio marginals for selected nested expert sets.}
\label{fig:dmm-sim-portfolio}
\end{figure}

\subsection{Individual approximation and credibility}

We select three fixed policyholder risks that differ substantially from the baseline portfolio: $u=-0.25$, $1.00$, and $1.75$, representing lower, intermediate, and high risk, respectively. For each risk level, we simulate 50 independent histories with lengths $T\in\{3,10,20,50,100\}$. Each displayed predictive curve is the posterior predictive distribution averaged across the 50 histories at the corresponding value of $T$. The true $L_1$ and $L_2$ distributions are obtained by deterministic numerical convolution on a fine lattice.

\Cref{tab:dmm-sim-individual} reports the mean KL divergence by risk level, response, and history length. \Cref{fig:low_risk,fig:mid_risk,fig:high_risk} compare the fitted posterior predictive distributions with their true counterparts. As more policyholder experience is observed, the predictive distributions move away from the common portfolio distribution at $T=0$ and toward the corresponding policyholder-specific generating distributions. This movement is reflected in the generally declining KL divergences for the directly modelled annual responses and their joint distribution.

The model recovers the policyholder-specific patterns reasonably well even from relatively short histories because the shared class composition is learned jointly from $N_1$, $L_1$, and $L_2$, and this updating also changes the implied distribution of $X_1$. The problem is especially difficult for $X_1$, which is not directly included in the fitted response, and for the aggregate losses. Claim occurrence is sparse, so even a long annual history may contain only a small number of years with positive claims and therefore limited direct information about severity. Learning of severity and aggregate-loss distributions is consequently slower than learning of frequency. The shared latent composition allows information from frequency and both aggregate losses to update all four predictive distributions simultaneously.

\begin{table}[htbp]
\centering
\caption{Mean posterior predictive KL divergence}
\label{tab:dmm-sim-individual}
\small
\begin{tabular}{llrrrrr}
\toprule
Risk level & $T$ & $N_1$ & $X_1$ & $L_1$ & $L_2$ & Joint\\
\midrule
Lower & 3 & 0.0175 & 0.5906 & 0.0357 & 0.0161 & 0.0478 \\
 & 10 & 0.0083 & 0.4080 & 0.0214 & 0.0082 & 0.0291 \\
 & 20 & 0.0043 & 0.3300 & 0.0151 & 0.0048 & 0.0208 \\
 & 50 & 0.0044 & 0.1444 & 0.0089 & 0.0042 & 0.0143 \\
 & 100 & 0.0048 & 0.1040 & 0.0079 & 0.0038 & 0.0130 \\
\addlinespace
Intermediate & 3 & 0.0534 & 0.1558 & 0.0814 & 0.0449 & 0.1513 \\
 & 10 & 0.0433 & 0.1320 & 0.0608 & 0.0374 & 0.1273 \\
 & 20 & 0.0224 & 0.1506 & 0.0420 & 0.0280 & 0.1012 \\
 & 50 & 0.0103 & 0.0800 & 0.0171 & 0.0170 & 0.0680 \\
 & 100 & 0.0078 & 0.0659 & 0.0124 & 0.0135 & 0.0594 \\
\addlinespace
High & 3 & 0.1220 & 0.0596 & 0.1473 & 0.1183 & 0.2719 \\
 & 10 & 0.0639 & 0.0499 & 0.0737 & 0.0685 & 0.1601 \\
 & 20 & 0.0486 & 0.0696 & 0.0656 & 0.0431 & 0.1296 \\
 & 50 & 0.0149 & 0.0695 & 0.0294 & 0.0300 & 0.0860 \\
 & 100 & 0.0085 & 0.0580 & 0.0194 & 0.0180 & 0.0631 \\
\bottomrule
\end{tabular}
\end{table}

\begin{figure}[H]
\centering
\includegraphics[width=.8\textwidth]{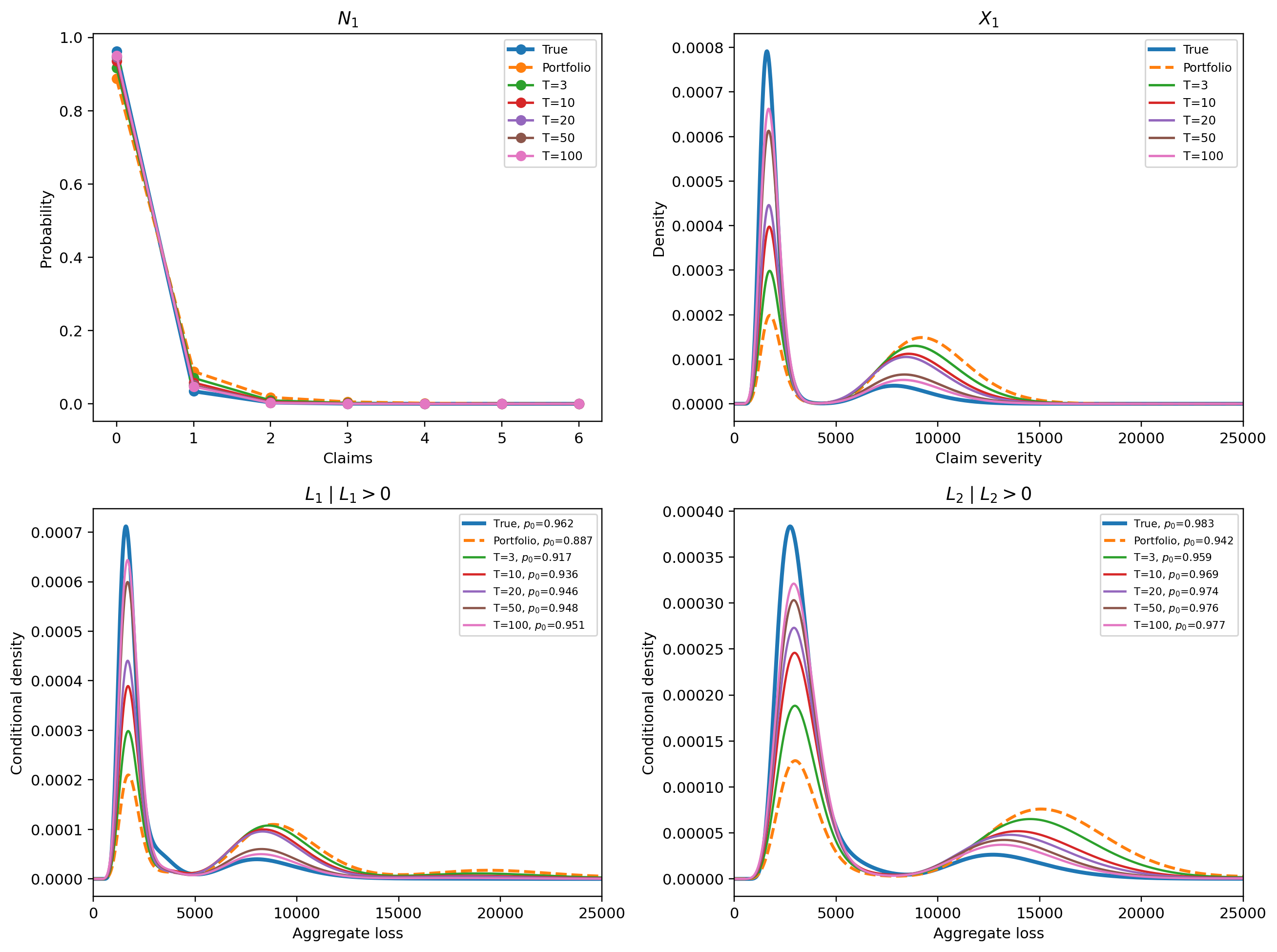}
\caption{Lower-risk policyholder ($u=-0.25$).}
\label{fig:low_risk}
\end{figure}

\begin{figure}[H]
\centering
\includegraphics[width=.8\textwidth]{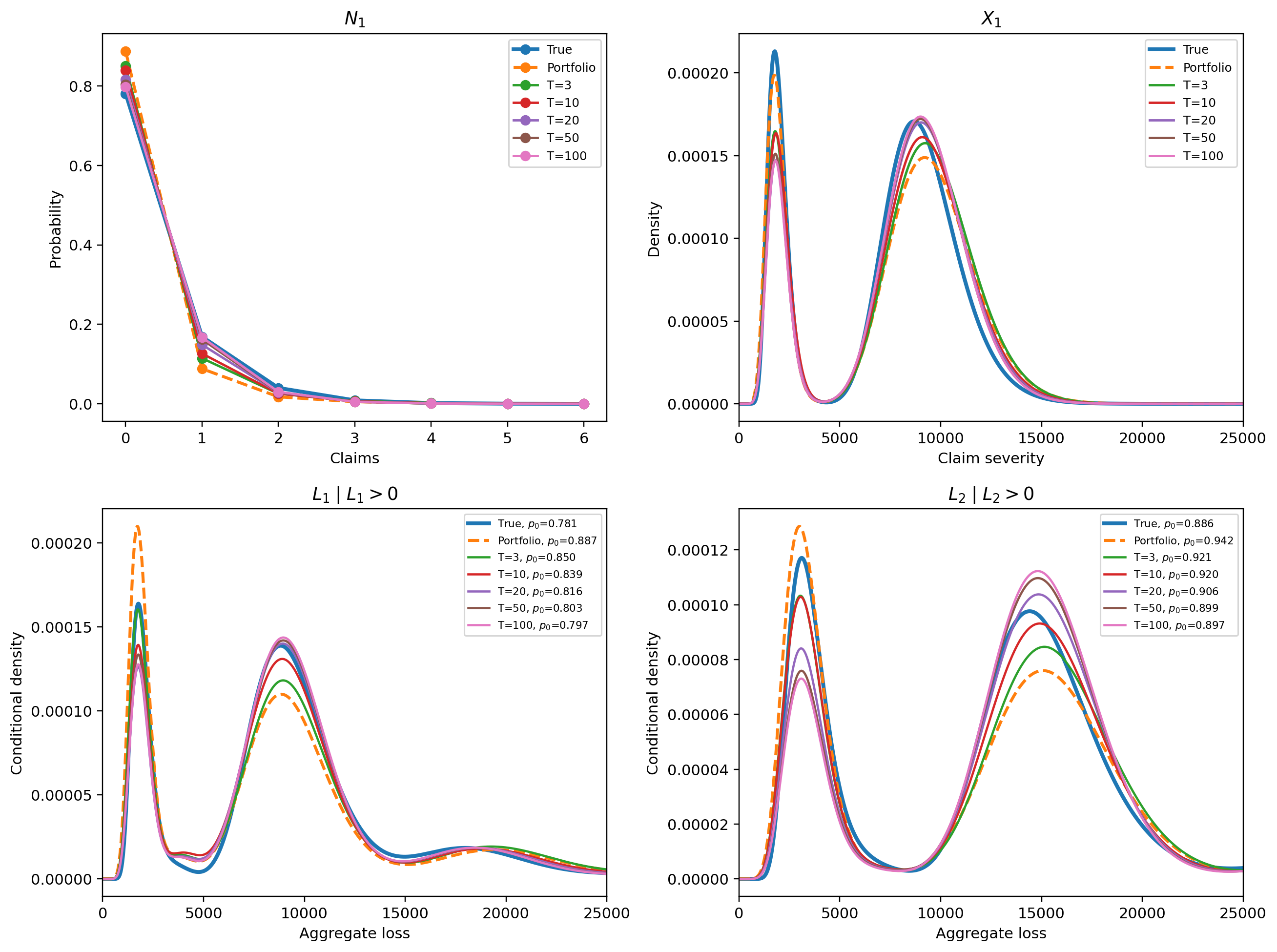}
\caption{Intermediate-risk policyholder ($u=1.00$).}
\label{fig:mid_risk}
\end{figure}

\begin{figure}[H]
\centering
\includegraphics[width=.8\textwidth]{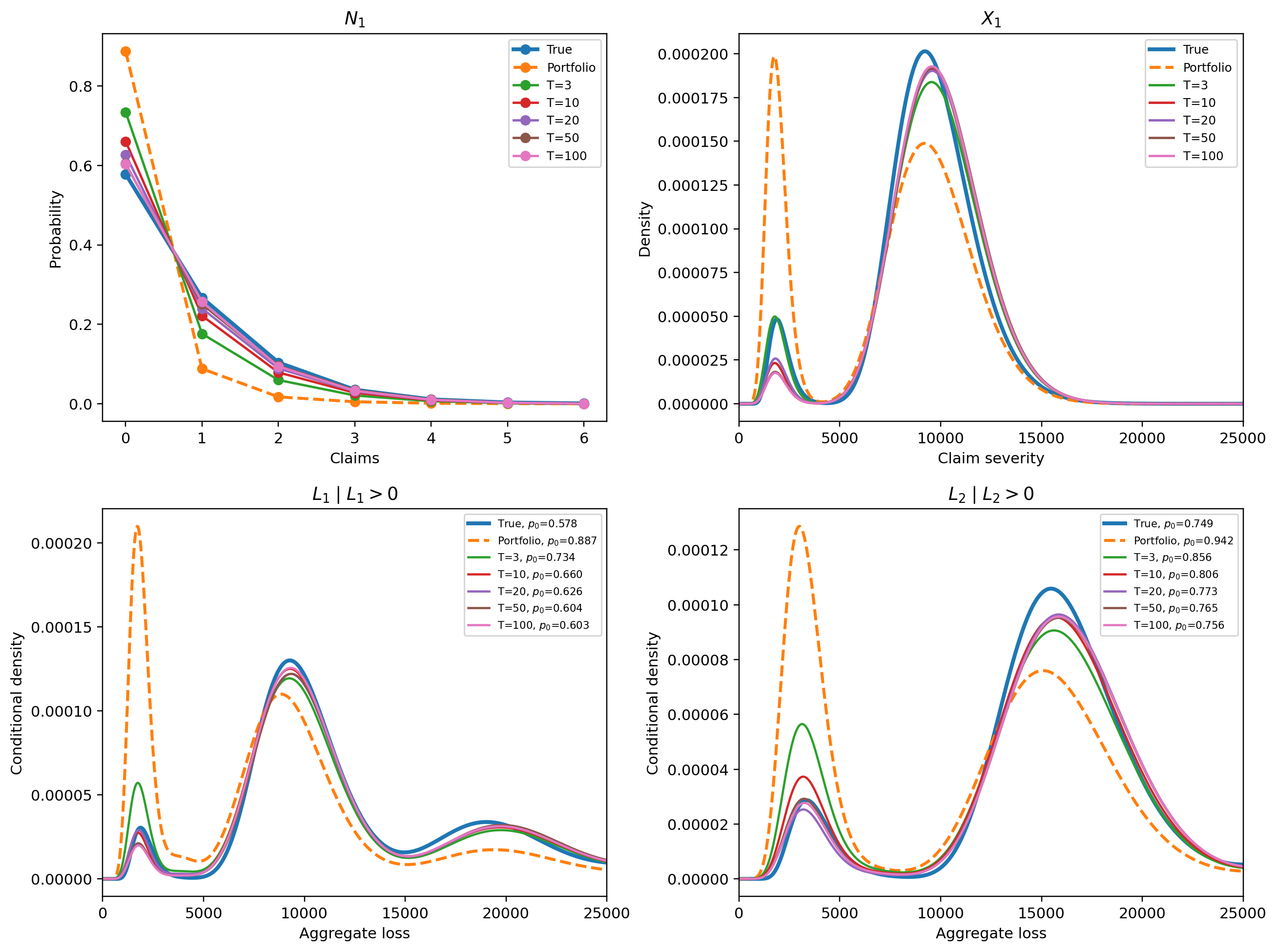}
\caption{High-risk policyholder ($u=1.75$).}
\label{fig:high_risk}
\end{figure}

Finally, \cref{tab:dmm-sim-premium} shows the evolution of the corresponding credibility premiums. For simplicity, we consider the pure premium principle, which is given by
\[
\mathbb{E}_{\widehat f_{i,T}}(L_1+L_2)
=
\sum_{k=1}^K
\widehat p_{ik,T}
\left(\lambda_k\eta_k\vartheta_k+\mu_k\right),
\]
where $\widehat p_{ik,T}$ is the posterior predictive class weight after $T$ periods. The premiums move clearly from the common portfolio premium toward the corresponding policyholder-specific targets as the policyholder history grows.

\begin{table}[htbp]
\centering
\caption{Mean predictive pure premium.}
\label{tab:dmm-sim-premium}
\small
\begin{tabular}{lrrrrrrr}
\toprule
Risk level & True & $T=0$ & $T=3$ & $T=10$ & $T=20$ & $T=50$ & $T=100$\\
\midrule
Lower & 204 & 1,978 & 982 & 722 & 333 & 239 & 224 \\
Intermediate & 3,585 & 1,978 & 2,482 & 2,547 & 3,157 & 3,357 & 3,471 \\
High & 10,919 & 1,978 & 7,101 & 9,052 & 9,884 & 10,630 & 10,693 \\
\bottomrule
\end{tabular}
\end{table}

\section{Numerical application: credibility and telematics}
\label{sec:application}

We consider an application of credibility in an increasingly important area of motor insurance: vehicle telematics. Telematics devices record how, when, and how much a vehicle is driven, including mileage, driving times, acceleration and braking, turning behaviour, and related patterns. Unlike traditional rating variables, these measurements accumulate during the policy period and provide a rich record of observed driving behaviour \citep{ChanEtAl2025}.

This creates a natural experience-rating problem. Claim history is informative but sparse: over short periods, most policyholders have no claims and relatively few have more than one. Telematics provides a much denser, although indirect, signal that may refine this limited experience. Recent studies show that behavioural information can improve risk classification, next-period claim prediction, and premium updating beyond conventional rating information \citep{DenuitGuillenTrufin2019,HenckaertsAntonio2022,ChanEtAl2025}. Its value, however, depends on how behavioural and claims information are combined and on the dependence structure allowed by the model. Despite the broader development of telematics ratemaking, explicitly credibility-based approaches remain relatively limited \citep{DenuitGuillenTrufin2019,CorradinEtAl2022}.

The application focuses on three questions. First, can the DMMM provide useful credibility updating under different degrees of behavioural persistence? Second, does telematics contribute information beyond conventional rating characteristics and a short claim history? Third, does the distributional formulation provide useful predictive information beyond the updated conditional mean? These questions allow us to study both the practical value of telematics as collateral experience and the flexibility of the proposed credibility structure. For simplicity and benchmarking, we focus on the pure premium principle.

Because a genuinely longitudinal public portfolio containing detailed claims and telematics is not available for this experiment, we construct controlled three-year pseudo-panels from the synthetic motor insurance portfolio of \citet{SoBoucherValdez2021}, using a data-driven procedure motivated by the SEAT framework of \citet{SoJeong2025} as we explain next.

\subsection{Data preprocessing}
\label{subsec:application_data}

The source portfolio released by \citet{SoBoucherValdez2021} contains $100{,}000$ annual policy records and 52 conceptual variables. These consist of 11 traditional policy and driver characteristics, 39 telematics variables, and two insurance responses. The complete set of variables, grouped according to their role in the present application, is summarised in \Cref{tab:application_variables}.

\begin{table}[htbp]
\centering
\caption{Variables in the source telematics portfolio and their use in the application.}
\label{tab:application_variables}
\scriptsize
\renewcommand{\arraystretch}{1.08}

\begin{tabularx}{0.96\textwidth}{
    >{\raggedright\arraybackslash}p{0.17\textwidth}
    >{\raggedright\arraybackslash}p{0.55\textwidth}
    >{\raggedright\arraybackslash}X
}
\toprule
\textbf{Group} & \textbf{Variables} & \textbf{Use} \\
\midrule

Traditional
(11)
&
\texttt{Duration};
\texttt{Insured.age}, \texttt{Insured.sex};
\texttt{Car.age}, \texttt{Car.use};
\texttt{Marital}, \texttt{Credit.score};
\texttt{Region}, \texttt{Territory};
\texttt{Annual.miles.drive};
\texttt{Years.noclaims}
&
Baseline rating and pseudo-panel matching.
\\

\addlinespace
Telematics:
driving allocation
(17)
&
\texttt{Annual.pct.driven}, \texttt{Total.miles.driven};
\texttt{Pct.drive.\{mon,tue,wed,thu,fri,sat,sun\}};
\texttt{Pct.drive.\{2hrs,3hrs,4hrs\}};
\texttt{Pct.drive.\{wkday,wkend\}};
\texttt{Pct.drive.rush\{am,pm\}};
\texttt{Avgdays.week}
&
Exposure, PCA, and behavioural matching.
\\

\addlinespace
Telematics:
driving intensity
(22)
&
\texttt{Accel.\{06,08,09,11,12,14\}miles};
\texttt{Brake.\{06,08,09,11,12,14\}miles};
\texttt{Left.turn.intensity\{08,09,10,11,12\}};
\texttt{Right.turn.intensity\{08,09,10,11,12\}}
&
PCA and behavioural matching.
\\

\addlinespace
Insurance
responses (2)
&
\texttt{NB\_Claim}, \texttt{AMT\_Claim}
&
Claim count is modelled; claim amount is not used.
\\
\midrule

\multicolumn{3}{l}{
\textit{Total: 52 variables = 11 traditional + 39 telematics + 2 insurance responses.}
}
\\

\bottomrule
\end{tabularx}
\end{table}

The present application focuses on claim frequency. Actual annual mileage is converted from miles to exposure in units of 100 kilometres and enters the claim-count models as an exposure offset. For PCA, \texttt{Total.miles.driven} is reserved for exposure, while \texttt{Pct.drive.sun} and \texttt{Pct.drive.wkend} are omitted as redundant compositional coordinates, leaving 36 telematics variables.

The baseline covariate vector $\mathbf x_i$ contains driver age and sex, vehicle age and use, and region. Driver and vehicle ages are standardised using the training-sample means and standard deviations; private vehicle use and rural region are the reference categories. The same $\mathbf x_i$ is used for all three annual observations of policyholder $i$. Realised mileage enters only through the claim exposure $e_{it}$. The longitudinal construction and the low-dimensional telematics representation are described in \cref{subsec:application_panel,subsec:application_pca}.

\subsubsection{Data-driven longitudinal pseudo-panel construction}
\label{subsec:application_panel}

The source portfolio is cross-sectional. In particular, \citet{SoJeong2025} identify longitudinality and serial dependence as natural extensions of the SEAT framework. To study credibility updating without imposing a parametric longitudinal data-generating model, we construct three-year histories by empirical donor matching.

A baseline donor determines the synthetic policyholder's initial traditional and behavioural characteristics. For subsequent years, eligible donors must match on sex, region, vehicle use, and age band, and are ranked using six traditional and 18 behavioural coordinates. The traditional block includes age, vehicle age, marital status, credit score, declared annual mileage, and territory, while the behavioural block includes realised mileage, driving-time shares, driving frequency, acceleration, braking, and turning measures. Claim count and claim amount are deliberately excluded from the matching rule.

Let $\Delta T_\ell$ and $\Delta B_\ell$ denote the transformed and scaled differences in the traditional and behavioural matching coordinates. The matching distance is
\begin{equation}
D^2
=
\frac{1}{6}\sum_{\ell=1}^{6}\Delta T_\ell^2
+
\frac{a_B}{18}\sum_{\ell=1}^{18}\Delta B_\ell^2.
\label{eq:application_donor_distance}
\end{equation}
For each renewal, one record is sampled uniformly from up to the 25 nearest eligible donors. A donor is not reused within the same synthetic policyholder, although it may appear in different synthetic policies.

Once a donor is selected, its telematics variables and insurance outcomes are inherited jointly. Thus, no model is used to generate claims or behavioural outcomes, and their within-year association is inherited from the empirical source records. Behavioural persistence is controlled through $a_B$: $a_B=1$ gives the traditional and behavioural matching blocks equal total weight, while $a_B=4$ places greater emphasis on similarity in the original behavioural characteristics. Because subsequent donors are still sampled from the nearest eligible records, the resulting histories retain substantial year-to-year variation. The two scenarios are therefore interpreted as controlled moderate- and high-persistence constructions rather than as panels calibrated to predetermined serial correlations.

Each scenario contains $15{,}000$ synthetic policyholders observed for three annual periods. Policyholders are randomly divided into training, validation, and test sets using a 70-15-15 split, corresponding to $10{,}500$, $2{,}250$, and $2{,}250$ policyholders, respectively. All three annual observations from a given policyholder remain in the same partition.

\subsubsection{Principal-component representation of telematics}
\label{subsec:application_pca}

The telematics block is high-dimensional and contains groups of strongly related measurements. We use principal component analysis (PCA) on the training set to obtain a small number of integrated behavioural dimensions rather than selecting a few individual event measures. Of the 39 telematics variables, we exclude total mileage because it is used separately as claim exposure, Sunday driving share as the reference for the day-of-week composition, and weekend driving share because it is the complement of weekday driving. The remaining 36 variables are standardised using the training-sample means and standard deviations. PCA is fitted separately under the moderate- and high-persistence constructions, and the estimated centering, scaling, and loadings are then fixed and applied unchanged to the corresponding validation and test policies.

The first three components explain $51.55\%$ of the standardised telematics variation under moderate persistence and $51.37\%$ under high persistence. There is also some levelling off after the third component: PC2 and PC3 each explain approximately $12$--$14\%$ of the variation, while the contribution falls to about $7$--$8\%$ for PC4 and PC5. We therefore retain the first three components as a parsimonious low-dimensional representation of the telematics information.

The retained components have clear behavioural interpretations. As shown in \cref{fig:application_pca_loadings}, PC1 is dominated by acceleration and braking intensities, PC2 primarily represents left-turn intensity, and PC3 primarily represents right-turn intensity. These descriptions refer to the variables with the largest loadings; each PC remains a linear combination of all 36 retained telematics variables. We therefore interpret the three PCs as integrated behavioural dimensions rather than individual driving measures. The loading patterns are highly similar across the two persistence scenarios.

\begin{figure}[htbp]
\centering
\includegraphics[width=\textwidth]{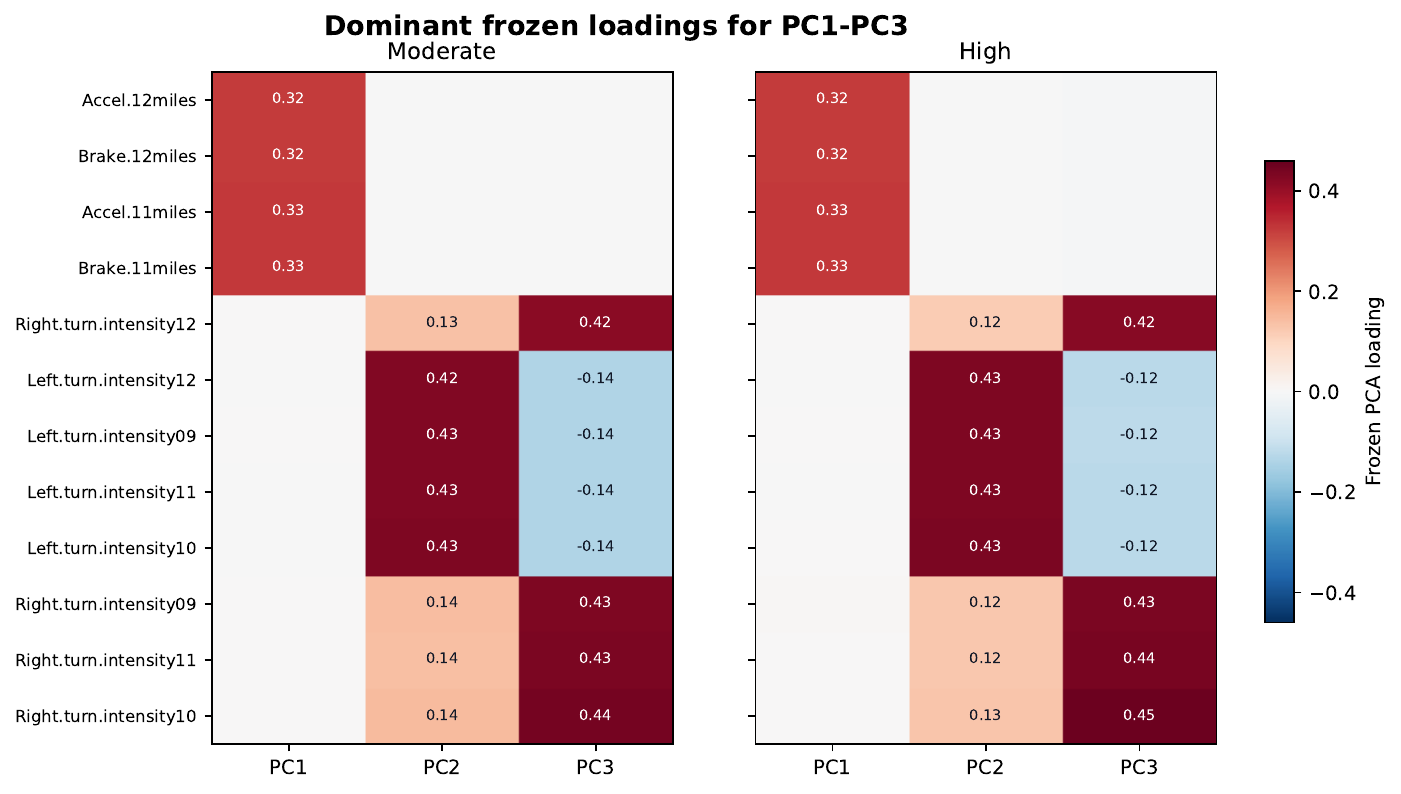}
\caption{Dominant loadings for the first three telematics principal components. Only variables with the largest absolute loadings are displayed.}
\label{fig:application_pca_loadings}
\end{figure}

Let $S_{it1},S_{it2},S_{it3}$ denote the three retained PC scores. The annual multivariate response is
\begin{equation}
\bY_{it}
=
\left(
N_{it},
S_{it1},
S_{it2},
S_{it3}
\right)^\trans,
\label{eq:application_response}
\end{equation}
where $N_{it}$ is annual claim count.

Because persistence is imposed through matching on the original behavioural characteristics rather than on the PC scores themselves, we also examine the serial dependence produced by the construction. \Cref{tab:application_persistence} reports pooled adjacent-year correlations for several representative original telematics variables and for the three retained PCs.

\begin{table}[htbp]
\centering
\caption{Empirical persistence of selected telematics variables and principal components.}
\label{tab:application_persistence}
\small
\begin{tabular}{llrrrr}
\toprule
& & \multicolumn{2}{c}{Moderate persistence} & \multicolumn{2}{c}{High persistence} \\
\cmidrule(lr){3-4}\cmidrule(lr){5-6}
Representation & Variable & Pearson & Spearman & Pearson & Spearman \\
\midrule
Original telematics & Acceleration (\texttt{Accel.09miles}) & 0.375 & 0.314 & 0.402 & 0.392 \\
& Braking (\texttt{Brake.09miles}) & 0.391 & 0.489 & 0.402 & 0.564 \\
& Rush-hour share (AM) & 0.427 & 0.389 & 0.513 & 0.464 \\
& Rush-hour share (PM) & 0.468 & 0.407 & 0.566 & 0.501 \\
\addlinespace
PCA scores & PC1: acceleration/braking & 0.326 & 0.609 & 0.343 & 0.698 \\
& PC2: left-turn intensity & 0.051 & 0.405 & 0.060 & 0.489 \\
& PC3: right-turn intensity & 0.010 & 0.388 & 0.268 & 0.486 \\
\bottomrule
\end{tabular}

\begin{minipage}{0.92\textwidth}
\footnotesize
\textit{Notes:} Entries are pooled adjacent-year correlations across policyholders, obtained by correlating the stacked year-1 and year-2 observations with the corresponding stacked year-2 and year-3 observations. The PCA scores are computed using the scenario-specific transformations estimated from the training samples and used in the subsequent analysis.
\end{minipage}
\end{table}

The selected original telematics measures are positively persistent under both constructions, with generally stronger dependence under the high-persistence design. This pattern extends beyond the variables displayed in the table: across all 36 telematics variables entering the PCA, the pooled adjacent-year correlation is larger under high persistence for 33 variables using Pearson correlation and for all 36 variables using Spearman rank correlation. The retained PCs preserve the same broad ordering, although the magnitude of linear dependence varies considerably across components. In particular, PC2 and PC3 show substantially stronger rank than linear persistence, consistent with the skewed and sparse threshold-event variables underlying these components. The empirical correlations therefore confirm that the two constructions differ meaningfully in longitudinal behavioural persistence without requiring that this difference translate uniformly to every principal component.

At the same time, the two pseudo-panels remain similar in their marginal characteristics. Annual claim frequency is approximately $0.072$ in both settings, with about $93.3\%$ of policy-year observations having no claims; the marginal distributions of exposure and the retained telematics components are likewise similar. Thus, the main distinction between the two scenarios is longitudinal dependence induced through the strength of behavioural matching, rather than deliberately different marginal portfolios. Full descriptive statistics are reported in \Cref{sec:telematics-stats}.

\subsection{Model specification}
\label{subsec:application_selection}

For the present application, we specify a DMMM in which each latent class has a Poisson claim-count expert and three Normal telematics experts:
\begin{align}
N_{it}\mid C_{it}=k
&\sim
\operatorname{Poisson}(e_{it}\lambda_k),
\label{eq:application_dir_claim}\\
S_{itj}\mid C_{it}=k
&\sim
\mathcal N(\mu_{jk},\sigma_{jk}^2),
\qquad j=1,2,3.
\label{eq:application_dir_pc}
\end{align}

The multinomial-logit gate defined in \cref{eq:softmax} represents the policyholder's a priori risk classification. Here $\mathbf x_i$ contains an intercept and the baseline variables listed above. The claim-count response is assigned weight $\omega_N=1$. The three PC responses share a total telematics weight $\omega_{\mathrm{tel}}$,
\[
\omega_{S_1}=\omega_{S_2}=\omega_{S_3}
=\frac{\omega_{\mathrm{tel}}}{3}.
\]
We set $\omega_{\mathrm{tel}}=1$ in the main analysis, so claim count and the complete telematics block receive equal total weight in class formation.

The DMMM is fitted using scale-aware direct random initialisations. The fitting algorithm is run for at most $10{,}000$ iterations, with early stopping when the absolute objective increment falls below $10^{-3}$. The number of latent classes is selected using the validation cohort only. We compare $K\in\{12,16,20,24\}$. \Cref{fig:application_K_selection} reports the full-history validation MSE and predictive log score.

\begin{figure}[htbp]
\centering
\includegraphics[width=\textwidth]{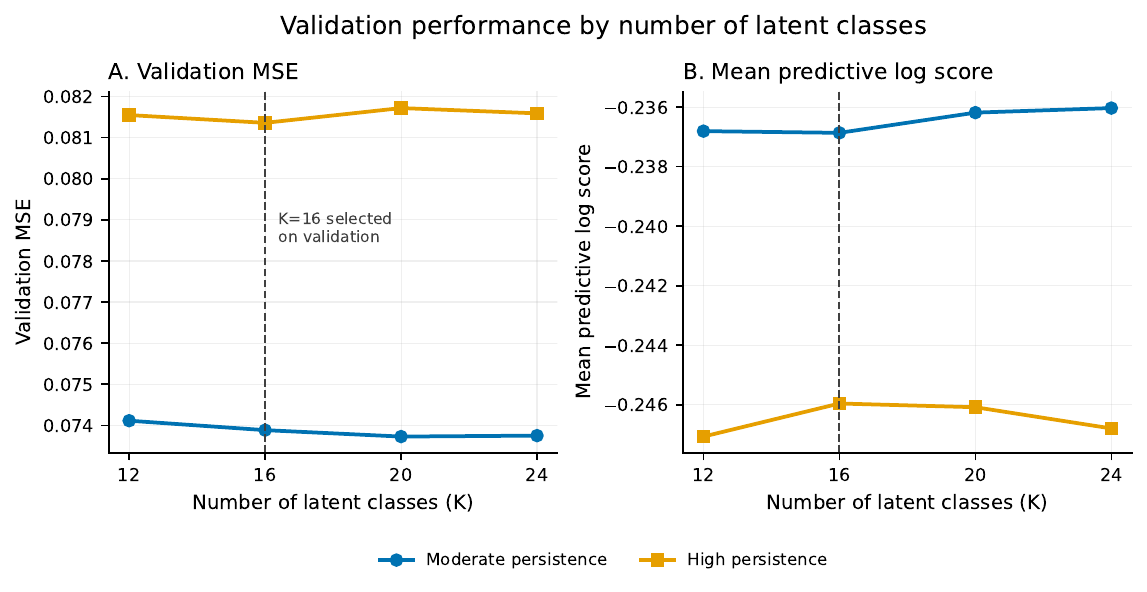}
\caption{Validation performance of the proposed model across the candidate
numbers of latent classes under the expanded fixed baseline specification.
Lower MSE and higher predictive log score indicate better validation
performance. The vertical line marks the validation-selected specification
$K=16$.}
\label{fig:application_K_selection}
\end{figure}

Under moderate persistence, validation performance is nearly indistinguishable for $K=16$, $20$, and $24$, whereas $K=16$ has the lowest MSE and highest log score under high persistence. We therefore fix $K=16$ before evaluating the test cohort, providing a parsimonious specification that is competitive in both settings. The classes are interpreted as a mixed-membership basis for heterogeneous risk--behaviour patterns rather than mutually exclusive driver types.

For comparison, we use the following benchmark models.

\begin{itemize}

\item \textbf{Common a priori benchmark.}

For comparison across methods, the same training-only Poisson generalized linear model (GLM) is used as a common a priori benchmark:
\begin{equation}
\log m_{it}^{(0)}
=
\log e_{it}
+
\mathbf x_i^\trans\boldsymbol\beta_0.
\label{eq:application_common_glm}
\end{equation}

The Poisson GLM therefore provides a method-neutral baseline for the cross-model comparison.

\item \textbf{Multivariate Gaussian random-effects model.}

The first benchmark is a multivariate Gaussian random-effects credibility model motivated by \citet{DenuitGuillenTrufin2019} and \citet{CorradinEtAl2022}. The same fixed baseline vector $\mathbf x_i$ is used as the fixed-effect design in all four margins. Let
\[
\mathbf b_i
=
(b_{iN},b_{i1},b_{i2},b_{i3})^\trans
\]
be a policyholder-specific random-effect vector. We specify
\begin{align}
N_{it}\mid b_{iN}
&\sim
\operatorname{Poisson}
\left[
e_{it}\exp\{\mathbf x_i^\trans\boldsymbol\beta_N+b_{iN}\}
\right],
\label{eq:application_gaussian_claim}\\
S_{itj}\mid b_{ij}
&\sim
\mathcal N
\left(
\mathbf x_i^\trans\boldsymbol\beta_j+b_{ij},
\sigma_j^2
\right),
\qquad j=1,2,3,
\label{eq:application_gaussian_pc}
\end{align}
with
\begin{equation}
\mathbf b_i
\sim
\mathcal N_4(\mathbf0,\boldsymbol\Sigma_b),
\label{eq:application_gaussian_re}
\end{equation}
where $\boldsymbol\Sigma_b$ is unrestricted. The model is fitted using a Laplace approximation. For individual prediction, the joint random-effect distribution is updated using the available history and integrated over the resulting posterior distribution.

\item \textbf{Regression-adjusted multivariate linear credibility.}

The second benchmark is a regression-adjusted multivariate linear credibility model, providing a direct multivariate analogue of classical B\"uhlmann credibility. We apply the usual second-moment credibility construction to residuals after removing the effects of the traditional rating variables. Claim count is represented by its Pearson residual from the baseline Poisson regression, while each of the three telematics PCs is represented by the residual from its corresponding linear regression. Let $\mathbf R_{it}$ denote the resulting four-dimensional residual vector.

From the training policies, we estimate the within-policy covariance $\widehat{\mathbf V}$ and the between-policy covariance $\widehat{\mathbf A}$ of the residuals using the usual multivariate credibility decomposition. For a set of observed history coordinates $H$, the claim-frequency residual adjustment based on the two-year history is
\begin{equation}
\widehat r_{i,N}(H)
=
\widehat{\mathbf A}_{N,H}
\left(
\widehat{\mathbf A}_{H,H}
+\frac{1}{2}\widehat{\mathbf V}_{H,H}
\right)^{-1}
\overline{\mathbf R}_{i,H}^{(2)},
\label{eq:application_linear_credibility}
\end{equation}
where $\overline{\mathbf R}_{i,H}^{(2)}$ is the mean residual vector over years 1 and 2. The year-3 claim prediction is then obtained by adding this adjustment to the baseline prediction,
\begin{equation}
\widehat N_{i3}(H)
=
\widehat m_{i3}^{(0)}
+
\sqrt{\widehat m_{i3}^{(0)}}\,
\widehat r_{i,N}(H).
\label{eq:application_linear_backtransform}
\end{equation}
Different choices of $H$ allow the same fitted model to use claims only, telematics only, or the complete multivariate history. When $H$ contains only the claim coordinate, the procedure reduces to the familiar scalar linear credibility update.

\end{itemize}

All population-level models are fitted using all three years of the training policies. The validation cohort is used for model specification and selection, most importantly to choose the number of latent classes $K$ for the proposed DMMM. For each validation policyholder, years 1 and 2 form the individual experience history and year 3 is reserved for next-period prediction. The test cohort remains untouched until all modelling choices are fixed and is then evaluated under the same design. Across all three credibility models, the a priori information, the two-year claim--telematics history, and the frozen three-PC representation are held fixed.

Because this application models claim frequency rather than claim severity, the conditional mean claim count is used as the primary rating target. Under a fixed expected severity, relative changes in this quantity translate directly into relative pure-premium adjustments. However, predictive performance is not the only consideration; the analysis also examines how much information is gained from individual experience, how telematics changes policyholder-level assessments, and what is learned from the full predictive claim-count distribution.

\subsection{What does individual experience add?}
\label{subsec:application_prediction}

We begin with the central credibility question: what is gained by observing a policyholder's own experience? The focus in this subsection is the updating mechanism within each fitted model, rather than the absolute ranking of models. For each test policyholder, four information sets are considered:
\begin{enumerate}[label=(\roman*),leftmargin=2em]
    \item baseline information only;
    \item years 1--2 claim counts only;
    \item years 1--2 telematics scores only; and
    \item the complete years 1--2 claim-count and telematics history.
\end{enumerate}
For the DMMM and Gaussian models, these are alternative information sets applied to the same fitted joint model; only the components supplied to the posterior update are changed, and the population model is not refitted. For linear credibility, the corresponding subset of residual coordinates is used in \cref{eq:application_linear_credibility}. The baseline-only prediction is therefore model-specific, but every model uses the same baseline covariates. The common Poisson GLM is used separately in the cross-model comparison in the next subsection.

Let
$\operatorname{MSE}_0$,
$\operatorname{MSE}_C$,
$\operatorname{MSE}_T$, and
$\operatorname{MSE}_F$
denote test MSE using no individual history, claims only, telematics only, and full history. Define
\begin{align}
G_C
&=
100\left(1-\frac{\operatorname{MSE}_C}{\operatorname{MSE}_0}\right),
&
G_T
&=
100\left(1-\frac{\operatorname{MSE}_T}{\operatorname{MSE}_0}\right),
\label{eq:application_gain_CT}
\\
G_{T\mid C}
&=
100\left(1-\frac{\operatorname{MSE}_F}{\operatorname{MSE}_C}\right),
&
G_F
&=
100\left(1-\frac{\operatorname{MSE}_F}{\operatorname{MSE}_0}\right).
\label{eq:application_gain_full}
\end{align}

\begin{table}[htbp]
\centering
\caption{Test-set full-history MSE and within-model credibility gains.}
\label{tab:application_history}
\small
\begin{tabular}{llrrrrr}
\toprule
Scenario & Model
& $\operatorname{MSE}_F$
& $G_C$
& $G_T$
& $G_{T\mid C}$
& $G_F$ \\
\midrule
Moderate
& Dirichlet, $K=16$
& \textbf{0.061748} & 0.56\% & 1.65\% & 1.20\% & 1.75\% \\
& Gaussian random effects
& 0.066346 & -5.46\% & 0.14\% & 0.03\% & -5.43\% \\
& Linear credibility
& 0.062752 & -0.44\% & 0.17\% & 0.11\% & -0.32\% \\
\addlinespace
High
& Dirichlet, $K=16$
& \textbf{0.074156} & 1.71\% & 1.39\% & 0.45\% & 2.16\% \\
& Gaussian random effects
& 0.076007 & 0.86\% & 0.08\% & 0.13\% & 0.99\% \\
& Linear credibility
& 0.074517 & 2.39\% & 0.12\% & 0.08\% & 2.46\% \\
\bottomrule
\end{tabular}
\end{table}

The gain columns in \cref{tab:application_history} are the main focus here. For the selected DMMM, all four credibility gains are positive in both test cohorts. Telematics alone reduces MSE by $1.65\%$ and $1.39\%$ under moderate and high persistence, respectively, and still adds information after claim history has been incorporated, with further reductions of $1.20\%$ and $0.45\%$. Using the complete history reduces MSE by $1.75\%$ and $2.16\%$ relative to the corresponding a priori predictions. These gains are modest, as expected from only two years of individual experience, but they are consistently positive for the proposed model.

The conventional benchmarks use the same information less consistently. Under moderate persistence, the claims-only updates of both benchmarks increase test MSE. Under high persistence, all three methods benefit from the complete history, but the incremental telematics-after-claims gains are only $0.13\%$ for the Gaussian model and $0.08\%$ for linear credibility. Thus, among the fitted specifications considered here, the DMMM is the only method for which claims-only, telematics-only, and full-history updates all improve MSE in both controlled persistence settings.

The two scenarios also show why behavioural persistence and incremental information should not be treated as the same thing. The full-history gain for the DMMM is larger under high persistence, but the additional gain from telematics conditional on claims is smaller. When claims and behaviour both reflect a persistent underlying risk, part of their information can be redundant. What matters for the present application is that telematics retains incremental value in both environments, even after the direct claim signal has been observed.

Before comparing the models directly, we check whether this updating pattern is tied to the validation-selected value of $K$. \Cref{tab:application_K_sensitivity} reports the test-set gains for all four fitted values of $K$; these results are used only as sensitivity checks and do not alter the choice of $K=16$.

\begin{table}[htbp]
\centering
\caption{Sensitivity of test-set credibility gains to the number of latent
classes. An asterisk denotes the specification selected using validation
before the test cohort was opened.}
\label{tab:application_K_sensitivity}
\small
\begin{tabular}{llrrrr}
\toprule
Scenario & $K$
& $G_C$
& $G_T$
& $G_{T\mid C}$
& $G_F$ \\
\midrule
Moderate
& 12 & 0.53\% & 1.60\% & 1.15\% & 1.67\% \\
& $16^\ast$ & 0.56\% & 1.65\% & 1.20\% & 1.75\% \\
& 20 & 0.31\% & 1.00\% & 0.76\% & 1.07\% \\
& 24 & -0.06\% & 1.56\% & 1.32\% & 1.27\% \\
\addlinespace
High
& 12 & 1.80\% & 1.57\% & 0.58\% & 2.37\% \\
& $16^\ast$ & 1.71\% & 1.39\% & 0.45\% & 2.16\% \\
& 20 & 1.70\% & 1.88\% & 1.16\% & 2.84\% \\
& 24 & 1.53\% & 1.46\% & 0.91\% & 2.42\% \\
\bottomrule
\end{tabular}
\end{table}

The telematics-only, telematics-after-claims, and full-history gains remain positive for every $K\in\{12,16,20,24\}$ in both settings. Their magnitudes vary with $K$, but the direction of the telematics contribution is stable. The conclusion that telematics provides useful collateral experience therefore does not depend on the single validation-selected number of classes.

\subsection{Predictive performance under the complete history}
\label{subsec:application_fit}

The previous subsection asks a within-model question: does an individual policyholder's history improve on that model's baseline-only prediction? We now hold the information set fixed at the complete two-year claim--telematics history and ask a different question: how do the fitted credibility methods compare with one another when each has access to the same information? These are the same held-out test policies used above, and no model is refitted or reselected at this stage. \Cref{tab:application_gof} reports MSE, MAE, mean Poisson deviance, predictive log score, and portfolio mean calibration. The calibration ratio is
\[
\frac{\overline{\widehat\mu}}{\overline N},
\]
where $\overline{\widehat\mu}$ and $\overline N$ are the mean predicted and observed year-3 claim counts. A ratio of one indicates exact calibration of the portfolio mean. The prediction error reduction (PER) is the percentage MSE reduction relative to the common Poisson benchmark.

\begin{table}[htbp]
\centering
\caption{Test performance under the complete two-year history.}
\label{tab:application_gof}
\small

\textbf{Panel A: Moderate persistence}

\vspace{0.25em}
\begin{tabular}{lrrrrrr}
\toprule
Model
& MSE
& MAE
& Mean deviance
& Log score
& Calibration
& PER (\%) \\
\midrule
Common Poisson
& 0.062550 & 0.117121 & 0.3186 & -0.214863 & 1.195 & -- \\
Dirichlet, $K=16$
& \textbf{0.061748} & 0.117970 & \textbf{0.3100}
& \textbf{-0.210579} & 1.234 & \textbf{1.28} \\
Gaussian random effects
& 0.066346 & 0.127183 & 0.3234 & -0.216257 & 1.423 & -6.07 \\
Linear credibility
& 0.062752 & \textbf{0.116956} & -- & -- & 1.206 & -0.32 \\
\bottomrule
\end{tabular}

\vspace{1em}
\textbf{Panel B: High persistence}

\vspace{0.25em}
\begin{tabular}{lrrrrrr}
\toprule
Model
& MSE
& MAE
& Mean deviance
& Log score
& Calibration
& PER (\%) \\
\midrule
Common Poisson
& 0.076399 & 0.132503 & 0.3724 & -0.256453 & 0.991 & -- \\
Dirichlet, $K=16$
& \textbf{0.074156} & 0.131650 & \textbf{0.3518}
& \textbf{-0.246619} & 1.023 & \textbf{2.94} \\
Gaussian random effects
& 0.076007 & 0.137961 & 0.3579 & -0.250346 & 1.164 & 0.51 \\
Linear credibility
& 0.074517 & \textbf{0.129890} & -- & -- & 0.987 & 2.46 \\
\bottomrule
\end{tabular}

\end{table}

When each credibility model is supplied with the complete two-year history, the selected DMMM has the lowest test MSE point estimate in both persistence settings, corresponding to PERs of $1.28\%$ and $2.94\%$ relative to the common Poisson model. It also has the lowest mean Poisson deviance and the best predictive log score among the probabilistic models. Linear credibility remains competitive for point prediction and has the lowest MAE in both settings, but it does not define a full predictive count distribution and produces three negative predictions in the moderate-persistence cohort. Thus, the proposed model's empirical advantage is not uniform dominance on every loss measure; rather, it combines competitive point prediction with a coherent probabilistic claim distribution and the more consistent multivariate updating seen above.

Calibration is close to one for the common Poisson, DMMM, and linear models under high persistence, while the Gaussian model overpredicts the portfolio mean. The moderate-persistence test cohort has a lower realised claim frequency than anticipated by all fitted models, including the common Poisson benchmark. The DMMM therefore overpredicts the portfolio mean by about $23\%$ in this cohort, although the Gaussian model overpredicts by more. This common portfolio-level shift should be kept in mind when interpreting the moderate-persistence predictive scores.

A distributional credibility model provides more than the updated conditional mean: the posterior class weights determine the complete next-period claim-count distribution. Because the DMMM and Gaussian models also provide full predictive distributions, we can compare more than their point predictions. \Cref{fig:application_predictive_distribution} compares the observed and fitted test-set probabilities of zero claims, one claim, and two or more claims. Under high persistence, the common Poisson benchmark closely matches the overall claim-count proportions, while the DMMM places slightly more mass on the multi-claim tail. Under moderate persistence, all three probabilistic models underpredict the zero-claim proportion and overpredict one-claim outcomes. The Gaussian model places substantially more mass on two or more claims ($1.01\%$ versus $0.44\%$ observed), whereas the DMMM is much closer at $0.53\%$.

\begin{figure}[htbp]
\centering
\includegraphics[width=\textwidth]{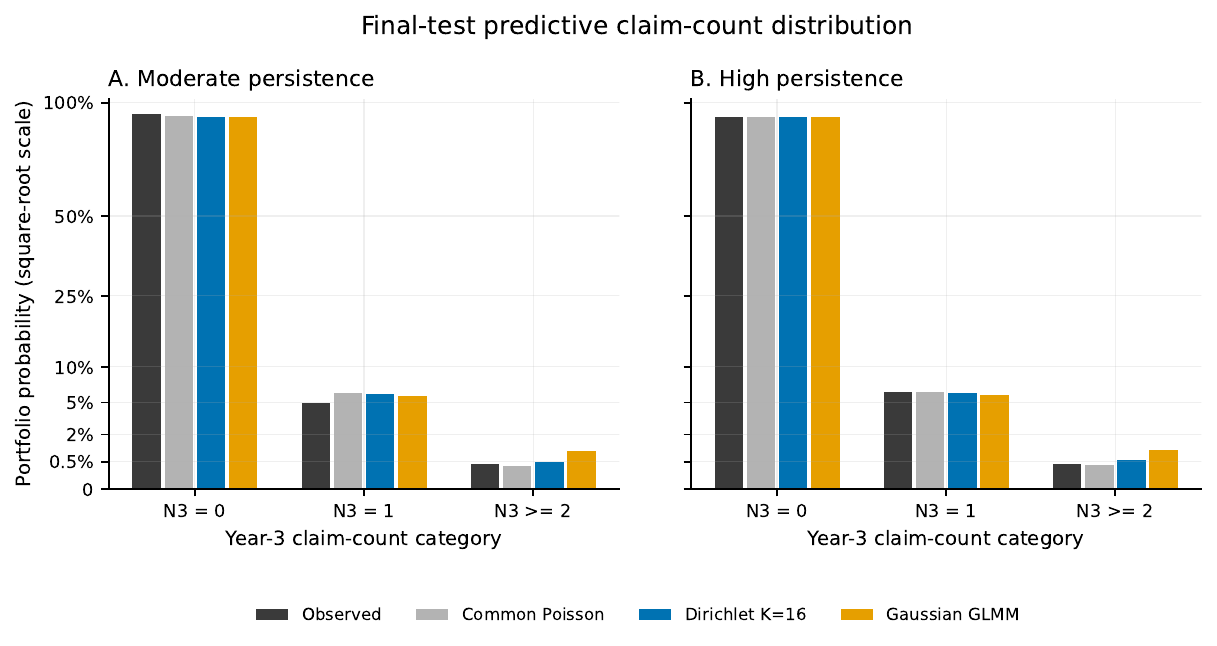}
\caption{Observed and fitted year-3 claim-count distributions in the final
test cohort after two years of history. The ordinate uses a square-root scale
to make the low-probability multi-claim category visible. The linear
credibility benchmark is omitted because it supplies a predictive mean rather
than a full predictive count distribution.}
\label{fig:application_predictive_distribution}
\end{figure}

\subsection{How experience changes claim-frequency assessments}
\label{subsec:application_corrections}

The preceding results show, first, that telematics adds incremental information within the proposed model and, second, that the resulting full-history predictions remain competitive with the benchmark methods. Portfolio-level scores, however, still do not show what the joint credibility update does to an individual policyholder. This matters because the intended role of telematics is not a common behavioural surcharge but additional differentiation among policies whose conventional claim histories may look similar. For policyholder $i$, let
\(
\widehat\mu_i^{(0)},
%\qquad
\widehat\mu_i^{(C)},
%\qquad
\widehat\mu_i^{(F)}
\)
denote the year-3 expected claim counts under the selected $K=16$ Dirichlet model using baseline information only, claim history only, and the complete claim--telematics history, respectively. Define
\begin{equation}
\mathrm{CF}_i^{(C)}
=
\frac{\widehat\mu_i^{(C)}}{\widehat\mu_i^{(0)}},
\qquad
\mathrm{CF}_i^{(F)}
=
\frac{\widehat\mu_i^{(F)}}{\widehat\mu_i^{(0)}}.
\label{eq:application_corrections}
\end{equation}

\Cref{fig:application_claim_full_corrections} displays these corrections in the final test cohort after grouping policyholders by claims observed during years 1 and 2. Claim history provides the first layer of a posteriori differentiation, while telematics creates additional spread within the same recent-claim group. The behavioural information does not act in one direction: depending on the observed history, it can attenuate or reinforce the claims-based adjustment. This is the policy-level implication of updating one shared latent risk composition with both experience sources.

\begin{figure}[htbp]
\centering
\includegraphics[width=\textwidth]{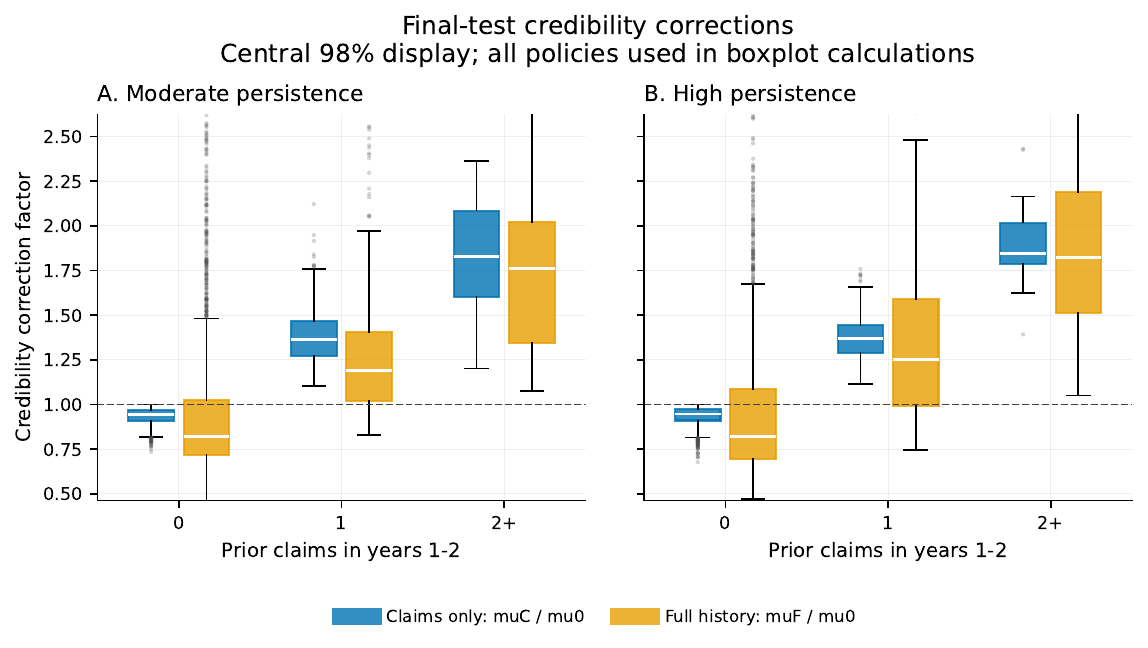}
\caption{A posteriori claim-frequency correction factors in the final test
cohort under the selected $K=16$, $\omega_{\mathrm{tel}}=1$ DMMM.
The claims-only correction is $\mathrm{CF}_i^{(C)}$, while the full-history correction
is $\mathrm{CF}_i^{(F)}$. Policyholders are grouped using claims observed in years 1--2
only. The horizontal reference line denotes no change from the a priori
prediction.}
\label{fig:application_claim_full_corrections}
\end{figure}

\subsubsection{Illustrative individual updates}

\Cref{fig:application_representative_policies} gives a policy-level view of the updating mechanism of the DMMM. The examples are selected using years 1--2 information only and are deliberately matched so that telematics, rather than different a priori risk levels, drives the visual contrast. Within the zero-claim and exactly-one-claim groups, we first identify policies in the lower and upper tails of the telematics-after-claims correction
$\widehat\mu_i^{(F)}/\widehat\mu_i^{(C)}$. A lower-tail and upper-tail policy are then paired to have similar a priori and claims-only predictions. Year-3 outcomes and prediction errors play no role in selecting the examples.

\begin{figure}[htbp]
\centering
\includegraphics[width=\textwidth]{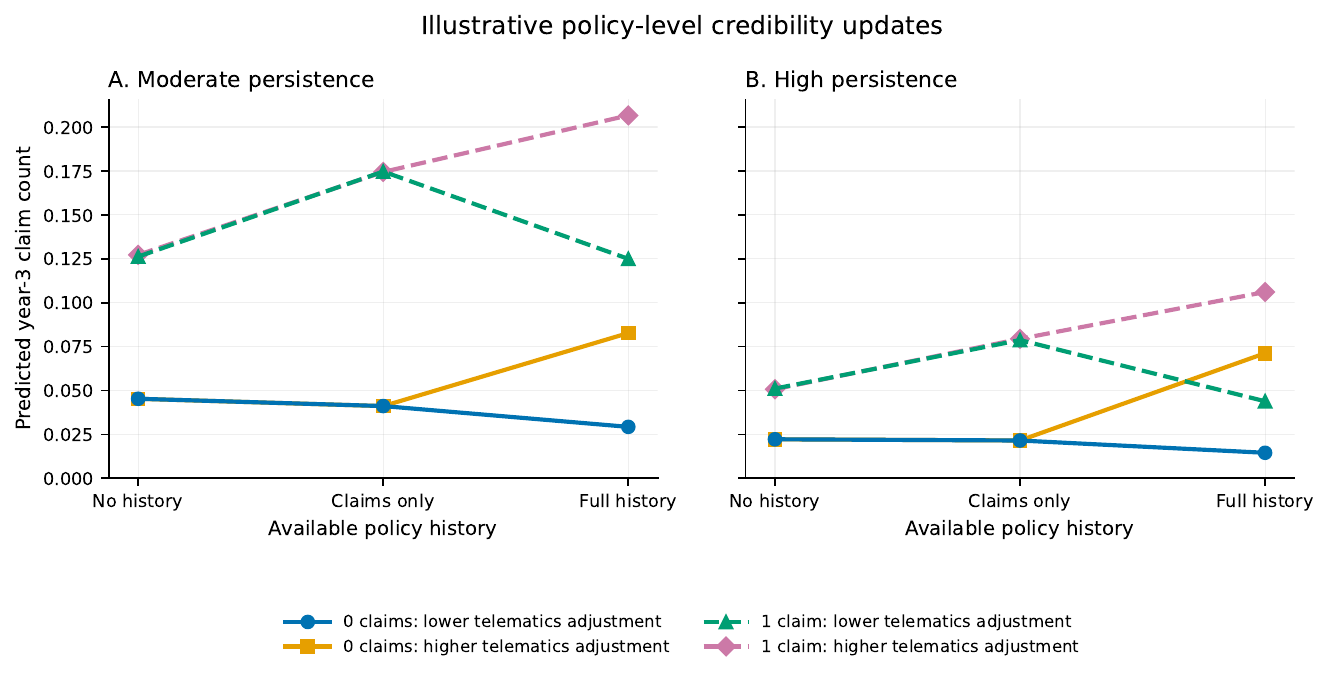}
\caption{Illustrative policy-level credibility updates in the final test
cohort under the selected $K=16$ DMMM. Within each persistence
setting, the zero-claim and one-claim pairs are chosen to have similar
a priori and claims-only predictions but contrasting
telematics-after-claims corrections. Selection uses only years 1--2
information; year-3 outcomes and prediction errors are not used.}
\label{fig:application_representative_policies}
\end{figure}

The matched examples make the incremental role of telematics explicit. Within each claims-history group, the paired policies begin with nearly the same a priori assessment and remain close after claims alone are incorporated. Once telematics is added, their predictions separate in opposite directions: one receives a downward behavioural adjustment and the other an upward adjustment. The pattern appears under both persistence settings and for both zero-claim and one-claim histories. Thus, the same recent claim experience can remain compatible with different latent risk--behaviour patterns, and the telematics history helps determine which classes receive greater posterior weight for the individual policy.

\section{Conclusion}\label{sec:conclusion}

Multivariate credibility is increasingly relevant as insurers observe several outcomes and auxiliary signals that may jointly reveal the underlying risk of a policyholder. This paper develops a static credibility model in which that latent risk is represented by a policyholder-specific Dirichlet composition over a finite collection of multivariate risk classes. The composition remains fixed over the observation horizon, while different periods may express different classes and the observed history progressively reveals the underlying composition. Baseline characteristics determine the a priori classification, and claims, losses, behavioural measures, or other outcomes jointly update the posterior composition. The shared latent structure also induces dependence across the response components without requiring a separate copula specification. Most importantly, the resulting posterior predictive distribution retains an exact credibility decomposition between portfolio and experience components. The expected-loss premium therefore retains the familiar credibility-weighted form, while the complete predictive distribution also provides the basis for quantiles, tail probabilities, nonlinear premium principles, and other actuarial functionals. An EM-type estimation procedure, with optional response weighting for auxiliary signals, provides a practical route for fitting the model.

The simulation study illustrates the flexibility behind this construction. Even when the fitted experts are deliberately simpler than the data-generating mechanism, richer collections of risk classes reproduce complex marginal behaviour and substantial dependence among frequency, severity, and losses. Repeated experience then progressively refines the policyholder-specific composition, illustrating how portfolio-level flexibility and individual credibility learning operate within the same model.

The telematics application provides the more practical assessment of the framework. Using an untouched test portfolio, the proposed model attained the lowest test MSE in both persistence settings and the best deviance and predictive log score among the probabilistic models, while linear credibility remained competitive for point prediction. More importantly, the information-ablation analysis showed that both claims and telematics contributed useful a posteriori information under the proposed model. Behavioural information continued to improve prediction after recent claims had already been incorporated, and this pattern remained stable across the fitted numbers of latent classes. At the individual level, telematics could either reinforce or attenuate a claims-based adjustment, allowing policyholders with similar baseline characteristics and the same recent claim count to receive different updated risk assessments. This illustrates the main practical role of the multivariate construction: richer policyholder experience can be incorporated within a familiar credibility framework rather than treated as a separate pricing mechanism.

Several extensions follow naturally. The present formulation is intentionally static, so a dynamic credibility model in which the latent composition evolves over time would be useful when longer histories contain evidence of behavioural change, recency, or regime transitions. Allowing the Dirichlet concentration to depend on policyholder characteristics would provide a second extension by permitting the strength of persistence and credibility to vary across risks. Richer within-class dependence structures and applications to genuinely longitudinal insurance portfolios are also natural directions. More broadly, the framework suggests that classical credibility can be retained even when the object being learned is a complex multivariate predictive distribution rather than a single premium or moment.

\section*{Declaration of competing interest}

The authors declare no conflict of interest.

\section*{Financial Support}

This research did not receive any specific grant from funding agencies
in the public, commercial, or not-for-profit sectors.

\section*{Declaration of generative AI and AI-assisted technologies in the manuscript preparation process}

During the preparation of this work, the authors used ChatGPT 5.6 Sol (OpenAI) to assist with language editing, improving clarity and organisation, and reviewing the presentation of mathematical and methodological arguments. After using this tool, the authors reviewed and edited the content as needed and take full responsibility for the content of the publication.

\section*{Data availability}

The simulation study uses data generated from the mechanism specified in \Cref{sec:dmm-simulation-study}. The telematics application uses the publicly available synthetic motor insurance portfolio of \citet{SoBoucherValdez2021}. The longitudinal pseudo-panels are constructed from that portfolio using the procedure described in \Cref{subsec:application_panel}; code for this construction is available with the accompanying Julia implementation.

\bibliographystyle{apalike}
\bibliography{main}

\appendix

\section{Computing the E-step}\label{subsec:computing_estep}

This appendix gives practical methods for computing the two quantities required by the weighted E-step in \cref{prop:EM_estep}. For policyholder $i$, these are
\[
\widetilde\tau_{itk}
=
\Prb_{\boldsymbol\omega}(C_{it}=k\mid \mathbf Y_i),
\qquad
\widetilde\xi_{ik}
=
\E_{\boldsymbol\omega}(\log p_{ik}\mid \mathbf Y_i),
\]
where the weighted expert score $\widetilde g_{itk}$ replaces the ordinary likelihood contribution $g_{itk}$. The calculations below therefore apply to both the ordinary likelihood and the weighted criterion. We suppress the policyholder index and write
\[
\alpha_k=\phi\pi_k,
\qquad
\alpha_0=\sum_{k=1}^K\alpha_k=\phi,
\]
with history length $T$.

Three implementations are useful in different regimes. For short histories, the \emph{grouping formula} is the simplest exact method: it replaces the $K^T$ class paths by a finite sum over partitions of the $T$ periods and gives both $\widetilde\tau$ and $\widetilde\xi$ by differentiation. For larger $T$, an exact \emph{count-state dynamic program} is preferable when the lattice of class-count vectors remains manageable; it avoids both path and set-partition enumeration. When both $T$ and $K$ make these exact calculations expensive, a collapsed \emph{Monte Carlo} E-step provides a scalable approximation.

\subsection{Closed-form calculation by grouping}\label{subsec:grouped_estep}

Let $\Pi_T$ be the set of partitions of $\{1,\ldots,T\}$. For each nonempty subset $A\subseteq\{1,\ldots,T\}$ define
\begin{equation}\label{eq:group_block}
h_A(\boldsymbol\alpha)
=
(|A|-1)!
\sum_{k=1}^K
\alpha_k
\prod_{t\in A}\widetilde g_{tk},
\qquad
\mathcal W_T(\boldsymbol\alpha)
=
\sum_{\mathcal A\in\Pi_T}
\prod_{A\in\mathcal A}h_A(\boldsymbol\alpha).
\end{equation}
Then the exact integrated weighted criterion is
\begin{equation}\label{eq:grouped_integrated_criterion}
\widetilde L_T
=
\frac{\mathcal W_T(\boldsymbol\alpha)}{(\alpha_0)_T},
\qquad
(\alpha_0)_T
=
\frac{\Gamma(\alpha_0+T)}{\Gamma(\alpha_0)}.
\end{equation}
Consequently,
\begin{align}
\widetilde\tau_{tk}
&=
\widetilde g_{tk}\,
\frac{\partial}{\partial \widetilde g_{tk}}
\log \mathcal W_T(\boldsymbol\alpha),
\label{eq:tau_derivative}\\
\widetilde\xi_k
&=
\psi(\alpha_k)-\psi(\alpha_0+T)
+
\frac{\partial}{\partial\alpha_k}
\log \mathcal W_T(\boldsymbol\alpha).
\label{eq:xi_normalizer_derivative}
\end{align}
Thus, the same finite polynomial gives both E-step quantities. The derivatives may be obtained analytically or by automatic differentiation. In particular,
\[
\frac{\partial h_A}{\partial\alpha_k}
=
(|A|-1)!\prod_{t\in A}\widetilde g_{tk},
\qquad
\frac{\partial h_A}{\partial\widetilde g_{tk}}
=
\mathbf 1_{\{t\in A\}}
(|A|-1)!\alpha_k
\prod_{u\in A\setminus\{t\}}\widetilde g_{uk},
\]
and derivatives of $\mathcal W_T$ follow by the product rule. If $\widetilde g_{tk}=0$, use the polynomial derivative directly rather than a logarithmic derivative.

For completeness, the identity follows from the Gamma representation of a Dirichlet vector. Let $G_k\sim\operatorname{Gamma}(\alpha_k,1)$ independently and set
\[
U_t=\sum_kG_k\widetilde g_{tk}.
\]
The joint cumulant for periods in $A$ is $h_A(\boldsymbol\alpha)$, so the moment--cumulant formula gives
\[
\E\left(\prod_tU_t\right)=\mathcal W_T(\boldsymbol\alpha).
\]
Since $(G_k/G_{+})_k$ is Dirichlet, is independent of
$G_{+}=\sum_kG_k$, and
$\E(G_{+}^T)=(\alpha_0)_T$,
\cref{eq:grouped_integrated_criterion} follows. Differentiation gives \cref{eq:tau_derivative,eq:xi_normalizer_derivative}.

This method is especially convenient for very short histories: the subset scores cost $O(K2^T)$ to construct and can be reused. Its limitation is the Bell-number growth of $\Pi_T$, so exactness does not imply scalability in $T$. The partition blocks are only a computational device and should not be interpreted as periods sharing the same realised class.

For very short histories, the partition polynomial can be written explicitly.
Let $\mathfrak P_{\mathbf a}([T])$ denote the set of partitions of
$[T]=\{1,\ldots,T\}$ whose block sizes are given by the integer partition
$\mathbf a$. With $h_A$ defined in \cref{eq:group_block},
\begin{align*}
\mathcal W_1
&=h_1,\\
\mathcal W_2
&=h_1h_2+h_{12},\\
\mathcal W_3
&=h_1h_2h_3
+\sum_{\mathcal A\in\mathfrak P_{(2,1)}([3])}
 \prod_{A\in\mathcal A}h_A
+h_{123},\\
\mathcal W_4
&=\prod_{t=1}^4h_t
+\sum_{\mathcal A\in\mathfrak P_{(2,1,1)}([4])}
 \prod_{A\in\mathcal A}h_A
+\sum_{\mathcal A\in\mathfrak P_{(2,2)}([4])}
 \prod_{A\in\mathcal A}h_A\\
&\quad
+\sum_{\mathcal A\in\mathfrak P_{(3,1)}([4])}
 \prod_{A\in\mathcal A}h_A
+h_{1234},\\
\mathcal W_5
&=\prod_{t=1}^5h_t
+\sum_{\mathbf a\in\{(2,1,1,1),(2,2,1),(3,1,1),
(3,2),(4,1)\}}
\ \sum_{\mathcal A\in\mathfrak P_{\mathbf a}([5])}
\prod_{A\in\mathcal A}h_A
+h_{12345}.
\end{align*}
Substitution into
\[
\widetilde\tau_{tk}
=
\widetilde g_{tk}
\frac{\partial}{\partial\widetilde g_{tk}}
\log\mathcal W_T,
\qquad
\widetilde\xi_k
=
\psi(\alpha_k)-\psi(\alpha_0+T)
+
\frac{\partial}{\partial\alpha_k}\log\mathcal W_T
\]
gives the corresponding exact closed-form E-step quantities for
$T\leq5$.

\subsection{Exact count-state dynamic programming}\label{subsec:count_dp}

A second exact approach groups partial class paths by their class-count vector. Let $F_t(\mathbf n)$ be the total integrated score of all paths through the first $t$ periods having counts $\mathbf n=(n_1,\ldots,n_K)$, where $|\mathbf n|=t$. With $F_0(\mathbf0)=1$,
\begin{equation}\label{eq:count_dp_forward}
F_t(\mathbf n)
=
\sum_{k:n_k>0}
F_{t-1}(\mathbf n-\mathbf e_k)
\frac{\alpha_k+n_k-1}{\alpha_0+t-1}
\widetilde g_{tk},
\qquad
\widetilde L_T
=
\sum_{|\mathbf n|=T}F_T(\mathbf n).
\end{equation}
For the backward pass set $R_T(\mathbf n)=1$ for $|\mathbf n|=T$ and, for $|\mathbf n|=t-1$,
\begin{equation}\label{eq:count_dp_backward}
R_{t-1}(\mathbf n)
=
\sum_{k=1}^K
\frac{\alpha_k+n_k}{\alpha_0+t-1}
\widetilde g_{tk}\,
R_t(\mathbf n+\mathbf e_k).
\end{equation}
The exact responsibilities are
\begin{equation}\label{eq:count_dp_tau}
\widetilde\tau_{tk}
=
\frac{1}{\widetilde L_T}
\sum_{|\mathbf n|=t-1}
F_{t-1}(\mathbf n)
\frac{\alpha_k+n_k}{\alpha_0+t-1}
\widetilde g_{tk}
R_t(\mathbf n+\mathbf e_k),
\end{equation}
while
\begin{equation}\label{eq:count_dp_xi}
\widetilde\xi_k
=
\sum_{|\mathbf n|=T}
\frac{F_T(\mathbf n)}{\widetilde L_T}
\left\{
\psi(\alpha_k+n_k)-\psi(\alpha_0+T)
\right\}.
\end{equation}

The recursion is exact because, after integrating out $\bp$, the next-class probability depends on the preceding path only through its counts. No additional multinomial coefficient is needed because the forward recursion has already summed over all compatible paths. Layer $t$ contains $\binom{t+K-1}{K-1}$ states. Hence this approach can be much cheaper than $K^T$ path enumeration when the count lattice remains moderate, but it also becomes expensive when both $T$ and $K$ are large. Scaled or log-space recursions should be used when underflow is a concern.

\subsection{Monte Carlo E-step}\label{subsec:mc_estep}

When the exact state spaces are too large, integrate out $\bp$ and sample the class sequence from its collapsed generalized posterior. A single-site Gibbs update is
\begin{equation}\label{eq:collapsed_gibbs}
\Prb_{\boldsymbol\omega}
(C_t=k\mid\mathbf C_{-t},\mathbf Y,\Theta)
\propto
\{\alpha_k+n_{k,-t}\}\widetilde g_{tk},
\end{equation}
where $n_{k,-t}$ excludes period $t$. For retained draws
$\mathbf C^{(1)},\ldots,\mathbf C^{(M)}$, with class counts $n_k^{(b)}$,
\begin{align}
\widetilde\tau_{tk}
&\approx
\frac{1}{M}
\sum_{b=1}^M
\mathbf 1\{C_t^{(b)}=k\},
\label{eq:MC_tau}\\
\widetilde\xi_k
&\approx
\frac{1}{M}
\sum_{b=1}^M
\left[
\psi\{\alpha_k+n_k^{(b)}\}
-
\psi(\alpha_0+T)
\right].
\label{eq:MC_xi}
\end{align}
These averages can be inserted directly into the EM/ECM updates. In Monte Carlo EM, simulation effort should increase as the algorithm approaches convergence, and Monte Carlo error should be monitored rather than equating the retained-draw count with effective information. This matters especially for small $\phi$, where the posterior may concentrate near almost-pure compositions and single-site Gibbs moves can mix slowly. Multiple dispersed chains, effective-sample-size checks, and block moves when needed are useful diagnostics; ascent-based MCEM gives a formal alternative for controlling Monte Carlo optimisation error \citep{CaffoJankJones2005}.

\newpage
\section{Descriptive statistics for the numerical application}
\label{sec:telematics-stats}

\begin{table}[htbp]
\centering
\caption{Descriptive statistics for the constructed three-year pseudo-panels.}
\label{tab:application_summary_appendix}
\small
\begin{tabular}{llrr}
\toprule
Variable & Statistic
& Moderate persistence
& High persistence \\
\midrule

Claim count
& Mean (SD)
& 0.0717 (0.2766)
& 0.0719 (0.2785) \\
& $N=0$ (\%)
& 93.3
& 93.3 \\
& $N=1$ (\%)
& 6.2
& 6.2 \\
& $N\geq2$ (\%)
& 0.5
& 0.5 \\
\addlinespace

Exposure (100 km)
& Mean (SD)
& 105.56 (87.44)
& 105.66 (87.86) \\
& Median [Q1, Q3]
& 83.76 [39.34, 150.98]
& 83.97 [39.56, 151.08] \\
\addlinespace

PC1
& Mean (SD)
& -0.017 (2.690)
& -0.016 (2.687) \\
& Median [Q1, Q3]
& -0.252 [-0.444, 0.054]
& -0.258 [-0.446, 0.050] \\

PC2
& Mean (SD)
& -0.007 (2.049)
& -0.003 (2.084) \\
& Median [Q1, Q3]
& -0.071 [-0.103, -0.021]
& -0.066 [-0.100, -0.013] \\

PC3
& Mean (SD)
& -0.004 (1.916)
& 0.000 (2.009) \\
& Median [Q1, Q3]
& -0.052 [-0.085, 0.002]
& -0.061 [-0.089, 0.000] \\
\addlinespace

Baseline driver age
& Mean (SD)
& 45.31 (13.75)
& 45.31 (13.75) \\
& Median [Q1, Q3]
& 45.00 [34.00, 55.00]
& 45.00 [34.00, 55.00] \\

Baseline vehicle age
& Mean (SD)
& 5.08 (4.01)
& 5.08 (4.01) \\
& Median [Q1, Q3]
& 4.00 [2.00, 8.00]
& 4.00 [2.00, 8.00] \\
\addlinespace

Male
& \%
& 59.6
& 59.6 \\
\addlinespace

Vehicle use: Private
& \%
& 17.6
& 17.6 \\
Vehicle use: Commute
& \%
& 50.9
& 50.9 \\
Vehicle use: Farmer
& \%
& 0.5
& 0.5 \\
Vehicle use: Commercial
& \%
& 30.9
& 30.9 \\

Urban
& \%
& 89.8
& 89.8 \\

\bottomrule
\end{tabular}

\vspace{0.4em}
\begin{minipage}{0.94\textwidth}
\footnotesize
\emph{Note:} Annual claim count, exposure, and principal-component scores are
summarised over 45,000 policy-year observations in each persistence setting.
Baseline characteristics are summarised once per policyholder
($15{,}000$ policyholders per setting). The baseline policyholders are common
to the two persistence settings, so their baseline summaries are identical.
The principal components are based on transformations estimated from the
training data only.
\end{minipage}
\end{table}

\end{document}